\documentclass[aps,pre,lengthcheck,superscriptaddress,nofootinbib,onecolumn,longbibliography]{revtex4-2}
\usepackage{amsmath,amssymb,amsfonts,physics,bigints,mathtools,amsthm}
\usepackage{graphicx}
\usepackage{appendix}
\usepackage[dvipsnames]{xcolor}
\usepackage{xurl}
\usepackage[colorlinks,allcolors=BrickRed,unicode,breaklinks=true]{hyperref}

\newcommand{\eq}[1]{\begin{align}#1\end{align}}
\newcommand{\bs}{\boldsymbol}
\newcommand{\mr}{\mathrm}

\newtheorem*{lemma}{Lemma}
\newtheorem*{corollary}{Corollary}

\begin{document}

\title{Linear stability and viscoelastic response in weakly jammed frictional granular systems}

\author{Masanari Shimada}

\email{m.shimada@mosk.tytlabs.co.jp}

\author{Yusuke Hara}

\author{Shihori Koyama}

\author{Ryohei Seto}

\author{Maria Yokota}

\email{yokota@mosk.tytlabs.co.jp}

\affiliation{Toyota Central R\&D Labs., Inc., Nagakute, Aichi 480-1192, Japan}

\begin{abstract}
Granular particles typically feature contact friction, which gives rise to various nontrivial properties of granular materials.
In this work, we investigate the linear stability of static packings composed of frictional spheres, whose interactions are modeled using the discrete element method.
We discuss a key difficulty in properly modeling static friction and address the problem of identifying rattlers.
We also explore the spectrum of the energy matrix, which governs the linear viscoelastic response of the system.
The intermediate- and high-frequency regimes of the spectrum remain qualitatively independent of both the friction coefficient and the packing fraction.
\end{abstract}

\maketitle

\section{Introduction}

Understanding and controlling friction remain major challenges in both fundamental physics and engineering~\cite{Brown1970Principles,deGennes1999Granular}.
In particular, granular materials such as sand and powder are typically composed of particles with frictional contacts and are widely used in many industrial applications.
Technologies that utilize these materials, however, still rely heavily on heuristic approaches because their behavior remains poorly understood in the field of soft condensed matter physics.
These mesoscopic granular particles are not subject to thermal fluctuations and do not reach thermal equilibrium.
Because of history-dependent contact friction, the state of the system cannot be defined solely by the current positions and momenta of the particles.
These properties of granular materials fall outside the conventional framework of statistical mechanics.

The jamming transition is central to our understanding of granular materials~\cite{Liu1998Jamming,OHern2003Jamming,Bi2011Jamming}.
When the packing fraction of a granular system is sufficiently low, the particles settle into configurations in which they do not contact one another.
In contrast, when the packing fraction exceeds a threshold value, known as the jamming transition point, the particles form an amorphous contact network that behaves like a solid; namely, the system supports a finite pressure and exhibits a finite shear modulus.
This non-equilibrium phase transition is unique to granular materials and has attracted much attention over the past few decades.

These jammed solids share some similarities with molecular glasses, another class of amorphous solids formed by quenching liquids, and considerable effort has been devoted to developing a unified framework for understanding a wide range of amorphous materials~\cite{Liu1998Jamming}.
With the development of various mean-field theories~\cite{Wyart2010Scaling,DeGiuli2014Effects,Franz2015Universal,Parisi2020Theory}, several important aspects of jammed materials have become clear.
These theories explain how particle packings reach jammed states and how the vibrational densities of states of jammed solids differ from those of their crystalline counterparts.

In parallel with these theoretical developments, numerical and experimental studies have also advanced significantly.
Linear stability analysis of jammed solids has established the existence of a plateau region in the intermediate-frequency regime of the vibrational density of states, which is followed on the lower-frequency side by the so-called non-Debye scaling region~\cite{Silbert2005Vibrations,Charbonneau2016Universal}.
The vibrational modes in these regimes are considered to be responsible for anomalous thermal and viscoelastic properties of amorphous materials.
These findings have also been examined experimentally using, for example, photoelastic disks~\cite{Majmudar2005Contact} and microrheological techniques~\cite{Hara2025Alink}.

Despite these achievements, almost all previous studies have been restricted to frictionless systems.
History-dependent contact friction fundamentally alters the properties of jammed packings.
Since it strongly damps particle motions, granular packings settle into mechanically stable yet significantly sparser amorphous structures than those formed by frictionless particles~\cite{Santos2020Granular}.
As a result, the jamming transition point in frictional systems is lower than that in frictionless systems.
Such sparse structures are characterized by the contact number per particle, which is six for frictionless packings at jamming but approaches four in the limit of strong friction.

To characterize the effects of contact friction on jammed packings, several studies have performed linear stability analyses of static packings of frictional particles in a manner similar to that used for frictionless systems~\cite{Chattoraj2019Oscillatory,Chattoraj2019Noise,Charan2020Transition,Bonfanti2020Oscillatory,Ishima2023Theory}.
The models employed in these studies are based on the discrete element method (DEM)~\cite{Cundall1979A}, in which history-dependent contact friction is represented by a virtual spring acting in the tangential plane at a contact.
This simplification enables large-scale simulations of frictional granular packings, and the DEM has been widely used in many areas of physics and engineering.
In particular, representing contact friction by a tangential spring makes the interparticle forces and torques differentiable, allowing one to compute their Jacobian matrix, which is the central quantity in linear stability analysis.
Refs.~\cite{Chattoraj2019Oscillatory,Chattoraj2019Noise,Bonfanti2020Oscillatory} demonstrated that frictional systems exhibit the so-called oscillatory instability under quasistatic shear in both two and three dimensions.
In Ref.~\cite{Ishima2023Theory}, the shear modulus of frictional granular packings was measured using a method originally developed for frictionless amorphous solids~\cite{Lemaitre2006Sum}.

The linear stability analysis of frictional granular packings has attracted growing interest in recent years, yet many aspects remain elusive.
In Refs.~\cite{Chattoraj2019Oscillatory,Chattoraj2019Noise,Bonfanti2020Oscillatory,Ishima2023Theory}, the initial mechanically stable configurations were prepared with frictional interactions artificially suppressed.
This protocol is useful for obtaining clean data that can be directly compared with results for frictionless systems. 
However, the resulting configurations differ substantially from experimentally generated granular packings, in which contact friction is inevitably present throughout the preparation process.
Given the strong dependence of the sparsity of contact networks on the friction coefficient~\cite{Santos2020Granular}, it is important to clarify how the friction coefficient affects linear stability of granular packings.
Furthermore, in materials engineering, the friction coefficient of granules is often tuned by coating them with another material~\cite{Shimamoto2022Adhesive}, making the friction coefficient a key design parameter in granular materials.

Motivated by these considerations, we investigate granular packings generated through a preparation process in which contact friction is present throughout the simulation.
Within the framework of linear stability analysis, we focus on how contact friction affects the time evolution of the mechanical energy.
After linearizing the equations of motion, the mechanical energy can be expressed as a quadratic form associated with a symmetric matrix.
As in frictionless systems, the spectrum of this matrix governs the linear viscoelastic response of the system~\cite{Lemaitre2006Sum,Tighe2011Relaxations,Hara2025Alink,Koyama2025Enhanced}.
Although the DEM has several limitations for studying low-frequency properties, the intermediate- and high-frequency regimes of the spectrum are found to depend only weakly on the friction coefficient and the packing fraction.

This paper is organized as follows.
In Sec.~\ref{sec: model and methods}, we introduce the model adopted in this study and the numerical protocol used to generate granular packings.
In Sec.~\ref{sec: theory}, we present the theoretical framework underlying the linear stability analysis of granular packings and discuss a challenge in examining the stability of frictional systems.
In Sec.~\ref{sec: numerical results}, we present our numerical results, based on which we discuss the definition of rattlers and the dependence of the density of states of the energy matrix on the friction coefficient.
In Sec.~\ref{sec: discussion}, we discuss several alternative approaches to modeling Coulomb friction without resorting to tangential springs.
In Sec.~\ref{sec: summary}, we summarize our findings and conclude the paper.
In Appendix~\ref{sec: mathematical preliminaries}, we introduce the notation and mathematical results used throughout this paper.
In Appendix~\ref{sec: tangential force}, we provide the complete definition of the tangential displacement required to describe the history-dependent tangential spring.
In Appendix~\ref{sec: derivatives of force and torque}, we present the detailed derivation of the Jacobian matrix of the forces and torques.
In Appendix~\ref{sec: minimal model}, we present simple examples illustrating the limitations of the DEM.
In Appendix~\ref{sec: alternative torque}, we provide numerical data demonstrating that the results presented in Sec.~\ref{sec: energy matrix and rattlers} are insensitive to the choice of the definition of pair torque.
In Appendix~\ref{sec: additional numerical results}, we provide additional numerical data and discuss the parameter and system size dependence of our results.

\section{Model and methods}\label{sec: model and methods}

Here we explain the general framework of the DEM~\cite{Cundall1979A}, introduce the model used in this study, and describe how we prepared mechanically stable reference configurations for the linear stability analysis.
We refer to these configurations as inherent structures following the convention of the study of molecular glasses~\cite{Widmer-Cooper2008Irreversible}.
Appendix~\ref{sec: notation} summarizes the mathematical notation used in this paper.

\subsection{Discrete element method}\label{sec: discrete element method}

Consider a system of $N$ spherical particles confined in a three-dimensional cubic box with periodic boundary conditions in all directions.
The mass of particle $i$ is denoted by $m_i$, where $i = 1,\ldots,N$.
The equations of motion are given by
\eq{
    \dv{\bs{r}_i}{t} & = \frac{\bs{p}_i}{m_i} , \label{position} \\
    \dv{\bs{p}_i}{t} & = \bs{F}_i + \bs{F}_{\mr{St},i} = \sum_{j(\neq i)} \bs{F}_{ij} + \bs{F}_{\mr{St},i} , \label{momentum} \\
    \dv{\bs{L}_i}{t} & = \bs{T}_i + \bs{T}_{\mr{St},i} = \sum_{j(\neq i)} \bs{T}_{ij} + \bs{T}_{\mr{St},i} , \label{angular momentum}
}
where $\bs{r}_i$, $\bs{p}_i$, and $\bs{L}_i$ are the position, momentum, and angular momentum of particle $i$, respectively.
The vectors $\bs{F}_i$ and $\bs{T}_i$ denote the internal force and torque acting on particle $i$, respectively, and are given by the sums of the pair forces $\bs{F}_{ij}$ and torques $\bs{T}_{ij}$.
Both $\bs{F}_{ij}$ and $\bs{T}_{ij}$ are zero if particles $i$ and $j$ are not in contact.
We also use the velocity $\bs{v}_i \coloneqq \bs{p}_i/m_i$ and the angular velocity $\bs{\omega}_i \coloneqq \bs{L}_i/J_i$, where $J_i \coloneqq (2/5)m_i R_i^2$ is the moment of inertia with the particle radius $R_i$.
Note that we assume a uniform mass distribution within a particle.
The vectors $\bs{F}_{\mr{St},i} \coloneqq - \gamma_v \bs{v}_i$ and $\bs{T}_{\mr{St},i} \coloneqq - \gamma_\omega \bs{\omega}_i$ are Stokes-type viscous interactions with damping coefficients $\gamma_v$ and $\gamma_\omega$, which are added to accelerate the convergence to inherent structures.
Without these single-particle dissipative interactions, dissipation occurs only when particles are in contact with each other. 
As a result, particles with fewer contacts dissipate energy more slowly, hindering the system from converging to an inherent structure within a reasonable timescale.

\begin{table}
    \renewcommand{\arraystretch}{1.2}
    \setlength\tabcolsep{0.5em}
    \centering
    \caption{
    Number of samples remaining jammed after removing rattlers and pseudo-rattlers for each combination of $\mu$, $\phi$, and $N$.
    The number of initial samples before removal is fixed at 10 for each parameter set.
    See Sec.~\ref{sec: energy matrix and rattlers} for details.
    }
    \begin{tabular}{cccc} \hline \hline
        $\mu$ & $\phi$ & $N=128$ & $N=1024$ \\ \hline
        0.00  & 0.690  & 10      & 10       \\ \hline
        0.01  & 0.650  & 10      & 10       \\
              & 0.670  & 10      & 10       \\
              & 0.690  & 10      & 10       \\ \hline
        0.20  & 0.610  & 10      & 10       \\
              & 0.630  & 10      & 10       \\
              & 0.650  & 10      & 10       \\ \hline 
        1.00  & 0.570  &  7      &  0       \\
              & 0.585  &  9      &  2       \\
              & 0.600  & 10      &  9       \\ \hline \hline
    \end{tabular}
    \label{tab:samples}
\end{table}

We next define the pair force and torque for particles $i$ and $j$ in contact.
The pair force $\bs{F}_{ij}$ is decomposed into two components as
\eq{
    \bs{F}_{ij} & = \bs{F}_{ij}^\parallel + \bs{F}_{ij}^\perp , \\
    \bs{F}_{ij}^\parallel & \coloneqq - f_{n,ij} (\Delta_{ij}) \hat{\bs{r}}_{ij} + \eta_{n,ij} (\Delta_{ij}) \mathcal P^\parallel (\bs{r}_{ij}) \cdot \bs{v}_{ij} , \label{F para} \\
    \bs{F}_{ij}^\perp & \coloneqq \min \qty ( \mu |\bs{F}_{ij}^\parallel|, | \bs{G}_{ij} | ) \hat{\bs{G}}_{ij} , \label{F perp}
}
where $\bs{r}_{ij} \coloneqq \bs{r}_j - \bs{r}_i$ and $\bs{v}_{ij} \coloneqq \bs{v}_j - \bs{v}_i$ are the relative position and velocity between the pair $ij$, respectively, and the normal force $f_{n,ij}$ and damping coefficient $\eta_{n,ij}$ are functions of the overlap $\Delta_{ij} \coloneqq R_i + R_j - r_{ij}\geq0$.
Eqs.~\eqref{F para} and \eqref{F perp} are parallel and perpendicular to the contact normal, respectively.
Eq.~\eqref{F perp} is the standard definition of the tangential force in the DEM while it is only a schematic expression to indicate that $\bs{F}_{ij}^\perp$ approximates the Coulomb friction with the friction coefficient $\mu$.
In particular, the switching from $| \bs{G}_{ij} |$ to $\mu |\bs{F}_{ij}^\parallel|$ corresponds to the transition from static to kinetic friction.
While we provide the complete definitions of $\bs{F}_{ij}^\perp$ and $\bs{G}_{ij}$ in Appendix~\ref{sec: tangential force}, which are necessary for the particle simulations described in the following subsection, we do not consider the transition from static to kinetic friction in the linear stability analysis and assume that all contacts remain sticking in Secs.~\ref{sec: theory} and \ref{sec: numerical results}.
Thus it suffices to consider static friction and Eq.~\eqref{F perp} simplifies to
\eq{
    \bs{F}_{ij}^\perp & = \bs{G}_{ij} \coloneqq k_{s,ij} ( \Delta_{ij} ) \bs{\xi}_{ij} + \eta_{s,ij} (\Delta_{ij}) \bs{v}_{ij}^\perp ,    
}
where $\bs{v}_{ij}^\perp \coloneqq \mathcal P^\perp \qty ( \bs{r}_{ij} ) \cdot \bs{v}_{ij} - \qty ( R_i \bs{\omega}_i + R_j \bs{\omega}_j ) \times \hat{\bs{r}}_{ij}$ is the tangential relative velocity, and $k_{s,ij}$ and $\eta_{s,ij}$ are the tangential spring constant and damping coefficient, respectively.
The relative tangential displacement $\bs{\xi}_{ij}$ is determined by the following equation of motion:
\eq{
    \dv{\bs{\xi}_{ij}}{t} & = \bs{v}_{ij}^\perp - \frac{\hat{\bs{r}}_{ij}\bs{\xi}_{ij}}{r_{ij}} \cdot \bs{v}_{ij} , \label{tangential eom}
}
where the second term in the right-hand side is interpreted as an inertial force~\cite{Silbert2001Granular}, see Appendix~\ref{sec: tangential force} for details.
Eq.~\eqref{tangential eom} is integrated while particles $i$ and $j$ are in contact with an initial condition $\bs{\xi}_{ij} = 0$ at the moment of the contact formation.
When the contact is lost, $\bs{\xi}_{ij}$ is reset to zero.
The tangential displacement is the only source of the history dependence in this model.
With this tangential force, the pair torque is written as
\eq{
    \bs{T}_{ij} & \coloneqq \qty ( R_i - \frac{\Delta_{ij}}{2} ) \hat{\bs{r}}_{ij} \times \bs{F}_{ij}^\perp . \label{pair torque}
}
Note that to realistically model Coulomb friction the standard DEM often includes rolling and twisting torques in addition to the tangential one in Eq.~\eqref{pair torque}~\cite{Luding2008Cohesive,Santos2020Granular}.
We neglect these interactions to construct a minimal model for investigating the effects of friction.

\subsection{Sample preparation}\label{sec: sample preparation}

Several variants of the interaction exist in the DEM depending on the choice of the four functions $f_{n,ij}$, $\eta_{n,ij}$, $k_{s,ij}$, and $\eta_{s,ij}$.
In this study we adopted the linear spring-dashpot model for both normal and tangential forces
\eq{
    f_{n,ij} \qty ( \Delta_{ij} ) & = k_n \Delta_{ij} , \label{normal spring} \\
    \eta_{n,ij} \qty ( \Delta_{ij} ) & = \eta_n m_{\mr{eff},ij} ,\\
    k_{s,ij} \qty ( \Delta_{ij} ) & = k_s , \\
    \eta_{s,ij} \qty ( \Delta_{ij} ) & = \eta_s m_{\mr{eff},ij} ,
}
where $k_n$, $\eta_n$, $k_s$, and $\eta_s$ are positive constants, and $m_{\mr{eff},ij} \coloneqq m_im_j/(m_i+m_j)$ is the effective mass.
We set $m_i = m$ for all $i$ while the radius $R_i$ was treated as a random variable drawn from the uniform distribution on $[R, 1.2R)$ to prevent crystallization.
In this paper, we choose the units of mass, length, and time such that $m = 2R = \sqrt{m/k_n} = 1$, and set $k_s = 1$, $\eta_v = \eta_\omega = 1$, and $\eta_n = \eta_s = 0.5$.
The choice of $k_n$, $\eta_n$, $k_s$, and $\eta_s$ is the same as in Ref.~\cite{Santos2020Granular}, which simplifies the linear stability analysis as explained below.
All simulations in this study were conducted using Large-scale Atomic/Molecular Massively Parallel Simulator (LAMMPS)~\cite{Plimpton1995Fast,Thompson2022LAMMPS}.

To obtain inherent structures, we generated random configurations of particles at a sufficiently small packing fraction $\phi_0 \approx 0.1$\footnote{
To roughly estimate the packing fraction at this stage, we assumed that all particles had a radius $0.6$ and the actual value of $\phi_0$ was slightly smaller than $0.1$.
}.
It is often the case that particles in such completely random configurations overlap substantially, generating unrealistically large forces and torques.
To remove these overlaps, the frictional interactions were temporarily suppressed, $\eta_n = k_s = \eta_s = 0$, which enabled us to define a potential energy using Eq.~\eqref{normal spring}.
We then removed the overlaps by minimizing this potential energy.
We then activated the full interaction introduced above with specific values of $\mu$ and isotropically compressed these non-overlapping initial configurations to specific packing fractions $\phi$ over a time interval of $50$.
The numbers of particles $N$, friction coefficients $\mu$, and packing fractions $\phi$ used in this study are listed in Table~\ref{tab:samples}.
Note that the range of friction coefficients considered here spans nearly the entire physically relevant regime, from almost frictionless to strongly frictional systems~\cite{Santos2020Granular}.
The systems at these values have similar pressures, enabling us to compare packings with different friction coefficients at comparable pressure levels, see Fig~\ref{fig:pressure}.
After compression, we let the system evolve without any perturbations and stopped the simulations when the absolute value of the maximal force and torque component is below $10^{-10}$.
Since particles in the final states had small velocities and angular velocities $\lesssim 10^{-7}$, we set them to zero.
The value $10^{-7}$ therefore provides a rough estimate of the numerical precision; namely, values below this threshold should be regarded as zero numerically.
For each set of $N$, $\mu$, and $\phi$, we prepared $10$ samples though some of them unjammed after removing rattlers and pseudo-rattlers as explained in Sec.~\ref{sec: numerical results}.
For comparison we prepared packings of frictionless particles at $\phi=0.69$ with the steepest descent dynamics described in Ref.~\cite{Shimada2024Instantaneous}.

\begin{figure}
    \centering
    \includegraphics[width=0.4\textwidth]{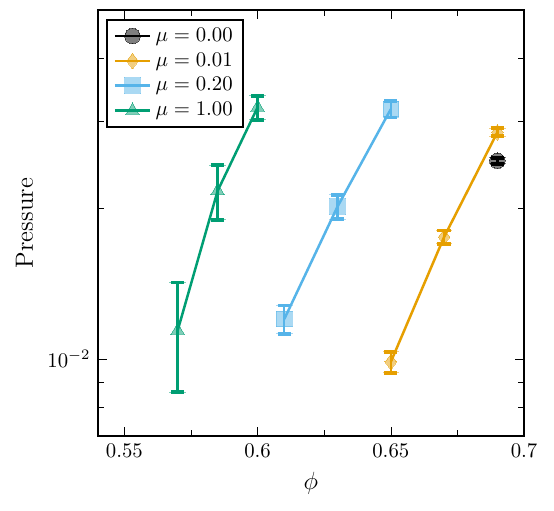}
    \caption{
    Pressure versus packing fraction for different friction coefficients.
    The system size is fixed at $N=128$.
    }
    \label{fig:pressure}
\end{figure}

\section{Theory}\label{sec: theory}

The theoretical analysis presented in this section is based on the interaction model introduced in Sec.~\ref{sec: discrete element method}.
The simple model introduced in Sec.~\ref{sec: sample preparation} is mainly intended for numerical simulations, and the theoretical structure discussed here does not rely on the specific choice of models.
Throughout the remainder of this paper, we strictly distinguish between the exact Coulomb friction law and the interaction model of the DEM.
In the former framework, the contact point between two particles remains fixed as long as the contact is in the static-friction regime. 
Once the transition to kinetic friction occurs, the contact begins to slip and the magnitude of the friction force is given by the kinetic friction coefficient multiplied by the normal contact force.
The latter, used in this study, is merely a simple model that approximates the former.

\subsection{Linearization}\label{sec: linearization}

The theoretical analysis begins with the linearization of the equations of motion.
Before proceeding to the detailed discussion, we define another state variable $\bs{\theta}_i$, which represents the orientation of particle $i$.
Since the orientation of a spherical particle does not affect the physical properties of the system, one does not need to consider the time evolution of the orientations in the previous section and Eq.~\eqref{angular momentum} is sufficient to describe the time evolution of the rotational degrees of freedom.
However, the \emph{relative} orientations of particles become relevant when analyzing deviations from a reference inherent structure and the time evolution of the orientation of particle $i$ is governed by
\eq{
    \dv{\bs{\theta}_i}{t} = \frac{\bs{L}_i}{J_i} . \label{orientation}
}
Note that Eq.~\eqref{orientation} is already an approximate form of the original equation of motion of the orientation under the assumption that the deviations from a reference orientation $\bs{\theta}_i = 0$ are infinitesimal, see Appendix~\ref{sec: orientation}.
From $\bs{\theta}_i$ and $\bs{\omega}_i$, we define the relative quantities for the rotational degrees of freedom
\eq{
    \bs{\chi}_{ij} & \coloneqq R_i \bs{\theta}_i + R_j \bs{\theta}_j , \\
    \bs{\nu}_{ij} & \coloneqq R_i \bs{\omega}_i + R_j \bs{\omega}_j .
}
Combining the four state variables $\bs{r}_i$, $\bs{\theta}_i$, $\bs{p}_i$, and $\bs{L}_i$, the full state vector measured from an inherent structure is given by
\eq{
    \bs{q} \coloneqq 
    \begin{pmatrix}
        \bs{r}^T - \bs{r}_0^T & \bs{\theta}^T & \bs{p}^T & \bs{L}^T
    \end{pmatrix}^T
    , \label{state vector}
}
where $\bs{r}_0$ is the positions of the particles at a given inherent structure.
The other state variables, $\bs{\theta}$, $\bs{p}$, $\bs{L}$, are zero at inherent structures.

Linearizing Eqs.~\eqref{position}--\eqref{angular momentum} and \eqref{orientation} around an inherent structure $\bs{q}=0$, we obtain
\eq{
    \dv{\bs{q}}{t} & = \mathcal A \cdot \bs{q} , \label{linear eom}
}
where the $12N$-dimensional Jacobian matrix $\mathcal A$ is given by
\eq{
    \mathcal A & \coloneqq 
    \begin{pmatrix}
        O_{6N} & M^{-1} \\
        H & \Gamma + \Gamma_{\mr{St}}
    \end{pmatrix} \label{Jacobian matrix}
}
with
\eq{
    M & \coloneqq 
    \begin{pmatrix}
        M_m    & O_{3N} \\
        O_{3N} & M_J
    \end{pmatrix}
    , \\
    H & \coloneqq 
    \begin{pmatrix}
        \pdv{\bs{F}}{\bs{r}} & \pdv{\bs{F}}{\bs{\theta}} \\[0.5em]
        \pdv{\bs{T}}{\bs{r}} & \pdv{\bs{T}}{\bs{\theta}} 
    \end{pmatrix}
    , \\
    \Gamma & \coloneqq 
    \begin{pmatrix}
        \pdv{\bs{F}}{\bs{p}} & \pdv{\bs{F}}{\bs{L}} \\[0.5em]
        \pdv{\bs{T}}{\bs{p}} & \pdv{\bs{T}}{\bs{L}} 
    \end{pmatrix}
    , \\
    \Gamma_{\mr{St}} & \coloneqq 
    \begin{pmatrix}
        \pdv{\bs{F}_{\mr{St}}}{\bs{p}} & O_{3N} \\[0.5em]
        O_{3N} & \pdv{\bs{T}_{\mr{St}}}{\bs{L}}         
    \end{pmatrix}
     = -
    \begin{pmatrix}
        \eta_v M_m^{-1} & O_{3N} \\
        O_{3N} & \eta_\omega M_J^{-1} 
    \end{pmatrix}
    .
}
The derivatives in $H$, $\Gamma$, and $\Gamma_{\mr{St}}$ are evaluated at a given inherent structure.
The matrix $-H$ is often called the dynamical matrix and central to the study of frictionless systems.
The mass matrix $M$ is decomposed into $3N$-dimensional matrices $M_m$ and $M_J$ defined as
\eq{
    M_m & \coloneqq
    \begin{pmatrix}
        m_1 I_3 && \\
        & \ddots & \\
        && m_N I_3
    \end{pmatrix}
    , \\
    M_J & \coloneqq
    \begin{pmatrix}
        J_1 I_3 && \\
        & \ddots & \\
        && J_N I_3
    \end{pmatrix}
    ,
}
where the off-diagonal elements are zero.

To further simplify the linearized equation of motion in Eq.~\eqref{linear eom} we change the state variable as
\eq{
    \bs{x} & \coloneqq \begin{pmatrix}
        M^{1/2} & O_{6N} \\
        O_{6N} & M^{-1/2}
    \end{pmatrix} \cdot \bs{q} . \label{change}
}
Eq.~\eqref{linear eom} is then rewritten as
\eq{
    \dv{\bs{x}}{t} & = \tilde{\mathcal A} \cdot \bs{x} , \label{normalized}
}
where
\eq{
    \tilde{\mathcal A} & \coloneqq 
    \begin{pmatrix}
        O_{6N} & I_{6N} \\
        \tilde H & \tilde \Gamma + \Gamma_{\mr{St}}
    \end{pmatrix} , \\
    \tilde H & \coloneqq M^{-1/2} \cdot H \cdot M^{-1/2} , \\
    \tilde \Gamma & \coloneqq M^{-1/2} \cdot \Gamma \cdot M^{1/2} .
}
Since $\Gamma_{\mr{St}}$ is diagonal, $M^{-1/2} \cdot \Gamma_{\mr{St}} \cdot M^{1/2} = \Gamma_{\mr{St}}$ holds.
We provide the explicit expressions of these matrices for the linear spring-dashpot model in Appendix~\ref{sec: derivatives of force and torque}.
Here we simply note that the matrix $H$, or equivalently $\tilde H$, is not symmetric, implying that the force field is non-conservative and therefore cannot generally be derived from a scalar potential.
The non-symmetric dynamical matrix arises from the asymmetry of the lever arm $\Delta_{ij} \neq \Delta_{ji}$ of the pair torque and the history-dependent tangential displacement $\bs{\xi}_{ij}$.
Although the former can be eliminated by changing the definition of the pair force or torque as in Refs.~\cite{Luding2008Cohesive,Chattoraj2019Oscillatory,Chattoraj2019Noise,Bonfanti2020Oscillatory} (see also Appendix~\ref{sec: minimal model} and \ref{sec: alternative torque}), the latter is inevitable as long as contact friction is modeled by a tangential spring.
When the system approaches jamming from above, both $\Delta_{ij}\to0$ and $\bs{\xi}_{ij}\to0$ hold and the dynamical matrix is exactly symmetric at jamming.
In addition, using the choice of the parameters in Sec.~\ref{sec: sample preparation}, the matrix $\Gamma$ becomes proportional to $H$ at jamming, which makes the following analysis quite simple~\cite{Koyama2025Enhanced}.

\subsection{Linear stability and mechanical energy}\label{sec: linear stability and mechanical energy}

If the real parts of all eigenvalues of the matrix $\tilde A$ are negative except for the trivial zero modes corresponding to the global translations, the linear dynamical system in Eq.~\eqref{normalized} is stable, meaning that the system returns to the original inherent structure $\bs{x} = 0$ starting from any initial conditions\footnote{
In this paper, the term ``stability'' refers to asymptotic stability and exponential stability, which are equivalent for linear autonomous systems~\cite{Khalil2001Nonlinear}.
}.
Since our system is closed and dissipative, the stability of the system amounts to the monotonically decreasing mechanical energy.
The time evolution of the total mechanical energy is given by
\eq{
    \dv{\tilde E(\bs{x})}{t} & = - \tilde R (\bs{x}) ,
}
where
\eq{
    \tilde E (\bs{x}) & \coloneqq \frac{1}{2} \bs{x} \cdot \tilde {\mathcal E} \cdot \bs{x} , \\
    \tilde R (\bs{x}) & \coloneqq \frac{1}{2} \bs{x} \cdot \tilde {\mathcal R} \cdot \bs{x} ,
}
with
\eq{
    \tilde{\mathcal E} & \coloneqq 
    \begin{pmatrix}
        - \tilde H_{\mr{sym}} & O_{6N} \\
        O_{6N} & I_{6N} 
    \end{pmatrix} , \\
    \tilde{\mathcal R} & \coloneqq
    \begin{pmatrix}
        O_{6N} & \tilde H_{\mr{skew}} \\
        - \tilde H_{\mr{skew}} & - 2 \qty ( \tilde \Gamma_{\mr{sym}} + \Gamma_{\mr{St}} )
    \end{pmatrix}
    .
}
In this paper, the matrices $\tilde{\mathcal E}$ and $\tilde{\mathcal R}$ are called the energy and dissipation matrices, respectively.
The blocks $-\tilde H_{\mr{sym}}$ and $I_{6N}$ in the energy matrix define the normalized potential and kinetic energies, respectively.
The block $\tilde \Gamma_{\mr{sym}} + \Gamma_{\mr{St}}$ in the dissipation matrix corresponds to the dissipation due to the dashpot interactions while $\tilde H_{\mr{skew}}$ results from the non-symmetric dynamical matrix.

\begin{figure}
    \centering
    \includegraphics[width=0.4\textwidth]{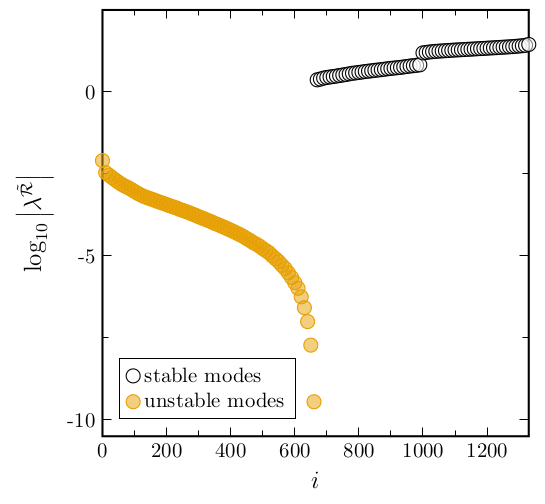}
    \caption{
    Eigenvalues of the dissipation matrix for a sample with $\mu=1.00$, $\phi=0.600$, and $N=128$. 
    The stable (unstable) modes satisfy $\lambda^{\tilde{\mathcal R}}_i > 0$ ($\lambda^{\tilde{\mathcal R}}_i < 0$).
    }
    \label{fig:R_eigvals}
\end{figure}

Mathematically it follows from the corollary in Appendix~\ref{sec: technical lemmas} that the dissipation matrix is indefinite as long as $\tilde H_{\mr{skew}} \neq O_{6N}$, and the system can spontaneously increase its energy via negative modes of the dissipation matrix.
This behavior is expected because non-symmetric dynamical matrices are characteristic of systems subject to non-conservative driving forces such as the Ziegler pendulum and the skew-symmetric part $\tilde H_{\mr{skew}}$ is proportional to an external driving force~\cite{Kirillov2013Nonconservative}. 
However, this is clearly at odds with the expectation that the DEM is a model of the Coulomb friction law since, as discussed above, it is required that $\tilde R(\bs{x})\geq0$ for all $\bs{x}$, which is equivalent to requiring that the dissipation matrix $\tilde{\mathcal R}$ is positive semidefinite.
Since the skew-symmetric part $\tilde H_{\mr{skew}}$ results from the history-dependent tangential displacement and the asymmetric lever arm of the pair force, both of which are peculiar to the DEM, this spurious instability is an artifact of the DEM.

This interpretation is also consistent with the phenomenon known as granular ratcheting~\cite{McNamara2008Microscopic} in granular packings under cyclic loading.
In Ref.~\cite{McNamara2008Microscopic}, it was shown that granular ratcheting can be eliminated by effectively removing the history dependence of the tangential spring.
This observation is closely related to the present instability: as long as the tangential spring retains its history dependence, $\tilde H_{\mr{skew}}$ remains nonzero, and negative modes of the dissipation matrix cannot in general be eliminated.

Fig.~\ref{fig:R_eigvals} illustrates these negative modes for a sample with $\mu=1.00$, $\phi=0.600$, and $N=128$ by showing the eigenvalues of the dissipation matrix as a function of their index.
The eigenvalues are sorted in ascending order $\lambda_1^{\tilde {\mathcal R}}<\lambda_2^{\tilde {\mathcal R}}< \cdots,$ and the stable (unstable) modes satisfy $\lambda^{\tilde{\mathcal R}}_i > 0$ ($\lambda^{\tilde{\mathcal R}}_i < 0$).
Fig.~\ref{fig:R_eigvals} indeed demonstrates that the matrix $\tilde R$ possesses both positive and negative eigenvalues.
Note that pseudo-rattlers are already removed in this sample, see Sec.~\ref{sec: numerical results}.

The spurious instability of the dissipation matrix renders the low-frequency, or low-energy, data obtained from the linear stability analysis unreliable even in the vicinity of jamming, where $\tilde H_{\mr{skew}}$ becomes small.
As shown in Ref.~\cite{Shimada2018Anomalous} using the Lennard-Jones (LJ) potential, the low-frequency modes are highly susceptible to small perturbations in the dynamical matrix.
It was shown that the low-frequency localized modes completely disappear due to the spurious instability resulting from the non-smooth LJ potential.
The spurious instability in this case is still manageable and it is possible to obtain the correct low-frequency modes by smoothing the LJ potential.
In the case of the DEM, however, there are no simple modifications to eliminate the spurious instability, to the best of our knowledge.
Since the low-frequency modes govern how the system converges to an inherent structure and how plastic events are triggered under shear, this instability hinders the investigation of such phenomena using the DEM.
In Appendix~\ref{sec: minimal model}, we provide simple examples that illustrate the spurious instability.

In light of these considerations, we focus on two topics in the following section.
After introducing the observables measured in our numerical analysis in Sec.~\ref{sec: observables}, we define rattlers and pseudo-rattlers in Sec.~\ref{sec: energy matrix and rattlers} and then explore the intermediate- and high-frequency regimes of the density of states of the energy matrix $\tilde{\mathcal E}$ in Sec.~\ref{sec: density of states}.
To investigate $\tilde {\mathcal E}$, we treat its only nontrivial block $\tilde H_{\mr{sym}}$, which we also refer to as the energy matrix.
In the study of frictionless systems, it is well known that the spectrum of $\tilde H_{\mathrm{sym}}$ determines the complex modulus and characterizes their linear viscoelasticity~\cite{Lemaitre2006Sum}.
Note that the dissipation matrix is also necessary for investigating the viscoelastic response in general.
However, since $\tilde \Gamma_{\mr{sym}}$ is approximately proportional to $\tilde H_{\mr{sym}}$ in the vicinity of jamming and $\tilde H_{\mr{skew}}$ only yields the spurious instability, we focus on $\tilde H_{\mr{sym}}$.
A detailed study of the quantitative differences between $\tilde \Gamma_{\mr{sym}}$ and $\tilde H_{\mr{sym}}$ is left for future work.

Finally, we briefly mention the stability of the linear dynamical system in Eq.~\eqref{normalized}.
Even though the mechanical energy $\tilde E(\bs{x})$ does not monotonically decrease due to the spurious instability, the system is still stable as long as the real parts of the eigenvalues of the matrix $\tilde {\mathcal A}$ are negative.
This stability is interpreted physically as follows: the system returns to the origin $\bs{x} = 0$ because the energy injection and dissipation due to the negative and positive modes of the dissipation matrix $\tilde R$ balance, as in the case of stationary states of externally driven systems.

\section{Numerical results}\label{sec: numerical results}

\begin{figure}
    \centering

    \begin{minipage}{0.32\linewidth}
        \centering
        \includegraphics[width=\linewidth]{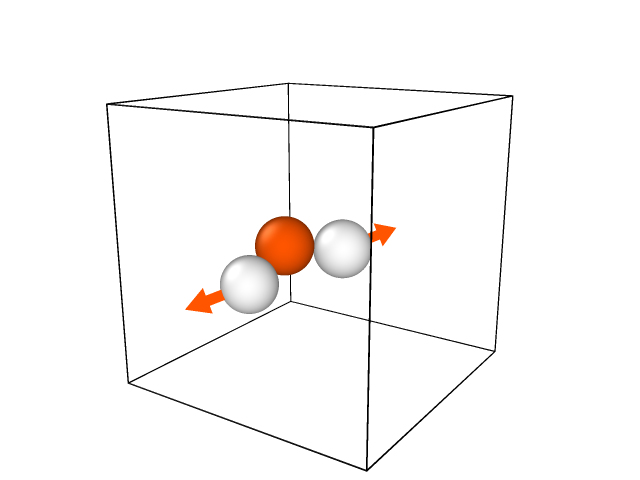}
    \end{minipage}
    \hspace{-0.05\linewidth}
    \begin{minipage}{0.32\linewidth}
        \centering
        \includegraphics[width=\linewidth]{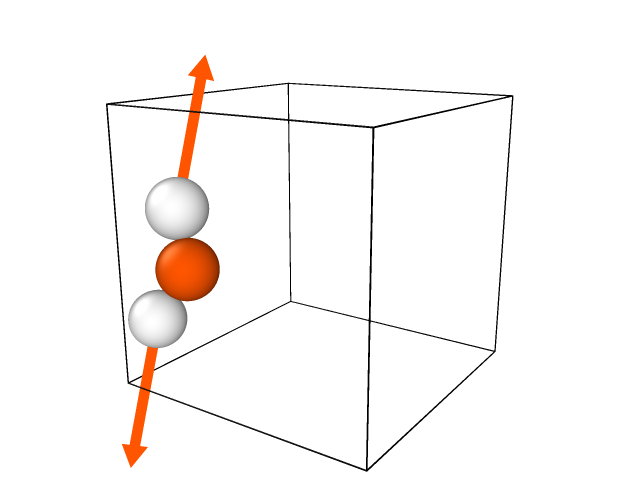}
    \end{minipage}
    \hspace{-0.05\linewidth}
    \begin{minipage}{0.32\linewidth}
        \centering
        \includegraphics[width=\linewidth]{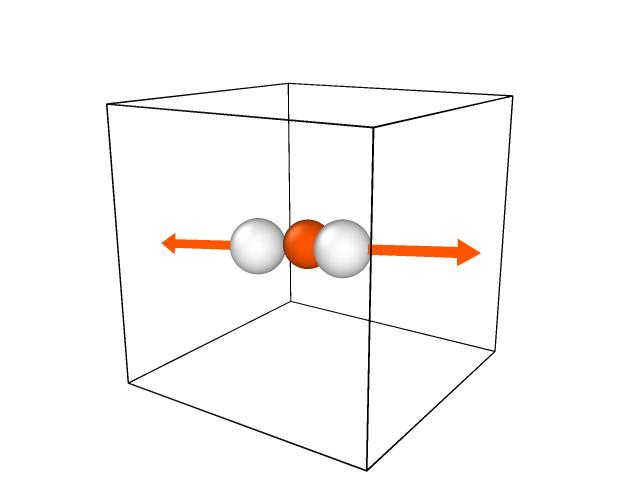}
    \end{minipage}
    \caption{
    Three pseudo-rattlers (orange) observed in a sample with $N=128$, $\mu=1.00$, and $\phi=0.600$.
    The orange arrows indicate the forces exerted by the pseudo-rattlers on their neighbor particles (white).
    The two forces in each panel have equal magnitude and opposite directions due to force balance.
    For visibility, the force vectors were magnified by factors of 5000, 100, and 2000 in the left, center, and right panels, respectively.
    }
    \label{fig:rattlers}
\end{figure}

\subsection{Observables}\label{sec: observables}

The density of states of the energy matrix $\tilde H_{\mr{sym}}$ is defined as
\eq{
    D (\omega) & \coloneqq \frac{1}{6N_\#-3} \sum_{i=1}^{6N_\#-3} \delta \qty ( \omega - \sqrt{|\lambda^{\tilde H_{\mr{sym}}}_i|} ) , \label{spectrum}
}
where $\delta$ is the Dirac delta function and $N_\#$ is the number of particles in a packing after removing rattlers and pseudo-rattlers as explained in the next subsection, and the factor of three corresponds to the trivial modes of global translations.
Also, since $\lambda^{\tilde H_{\mr{sym}}}_i<0$ for all $i$ after removing rattlers and pseudo-rattlers, we use the eigenfrequency $\omega_i \coloneqq \sqrt{|\lambda_i^{\tilde H_{\mr{sym}}}|}$, following the convention of the study of frictionless systems~\cite{Lerner2021Low-energy}.

We split each eigenvector $\bs{e}_i^{\tilde H_{\mr{sym}}}$ into two halves corresponding to the translational and rotational degrees of freedom
\eq{
    \bs{e}_i^{\tilde H_{\mr{sym}}} & \eqqcolon \begin{pmatrix}
        \bs{e}_i^{(r)} \\
        \bs{e}_i^{(\theta)}
    \end{pmatrix}
    .
}
The weight of each component is defined as $W_i^{(a)} \coloneqq |\bs{e}_i^{(a)}|^2$, which satisfies $W_i^{(r)} + W_i^{(\theta)} = 1$ by normalization. 
The participation ratio~\cite{Schober1991Localized} is defined as
\eq{
    P_i^{(a)} \coloneqq \frac{\qty ( \sum_{j=1}^{N_\#} | \bs{e}_{i,j}^{(a)} |^2 )}{N_\# \sum_{j=1}^{N_\#} | \bs{e}_{i,j}^{(a)} |^4} ,
}
where $a = r, \theta$.
It quantifies the degree of localization of a vector, and equals to one if the vector is completely extended over the system and $1/N_\#$ if the vector is localized on a single particle.

Given a matrix $A$ and its eigenvector $\bs{e}_i^A$, we also define the contribution of particles with a contact number $z$ to $\bs{e}_i^A$
\eq{
    w_i^A (z) \coloneqq \sum_{j=1}^{N_\#} \delta_{z_j,z} | \bs{e}_{i,j}^A |^2 ,
}
where $z_j$ is the contact number of particle $j$.
In the following subsection, we use $w_i^A(2)$ for $A = \tilde{\mathcal A}, \tilde{\mathcal R}$, and $\tilde H_{\mr{sym}}$ to discuss instability associated with particles with two contacts.

\subsection{Rattlers and pseudo-rattlers}\label{sec: energy matrix and rattlers}

Rattlers are particles that do not participate in the rigid structure of a jammed packing and whose removal does not affect the stability of the remaining system.
Therefore, particles with no contacts are obviously rattlers. 
Particles with only a single contact are also rattlers.
Due to the force and torque balance, such particles merely touch their neighbor without being constrained by the rigid structure of the system.

The situation differs between frictionless and frictional systems for particles with two contacts. 
While they are still rattlers in the case of frictionless systems, particles with two contacts are generally unable to move or rotate freely in any direction under the exact Coulomb friction law, i.e., they are mechanically constrained.
This is simply understood by counting the number of degrees of freedom and constraints exerted by contacts~\cite{Santos2020Granular}.
In fact, even in our simulations using the DEM, we observed several particles with two contacts that carry finite force and torque, see Fig.~\ref{fig:rattlers}.
Hereafter, we refer to particles with two contacts as pseudo-rattlers to distinguish them from true rattlers.
The numerical data in this subsection were obtained after removing only true rattlers.

\begin{figure}
    \centering

    \begin{minipage}{0.32\linewidth}
        \centering
        \includegraphics[width=\linewidth]{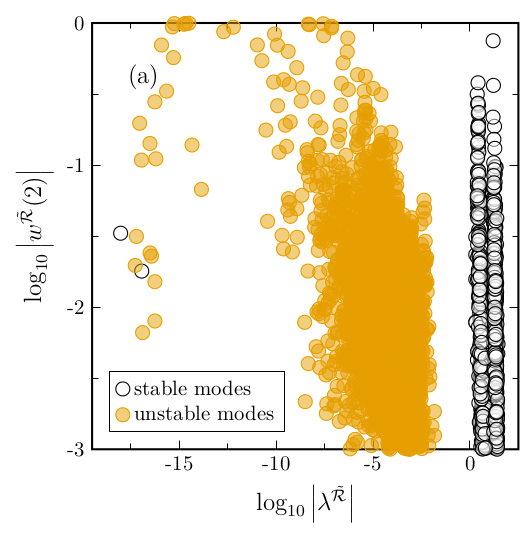}
    \end{minipage}
    \begin{minipage}{0.32\linewidth}
        \centering
        \includegraphics[width=\linewidth]{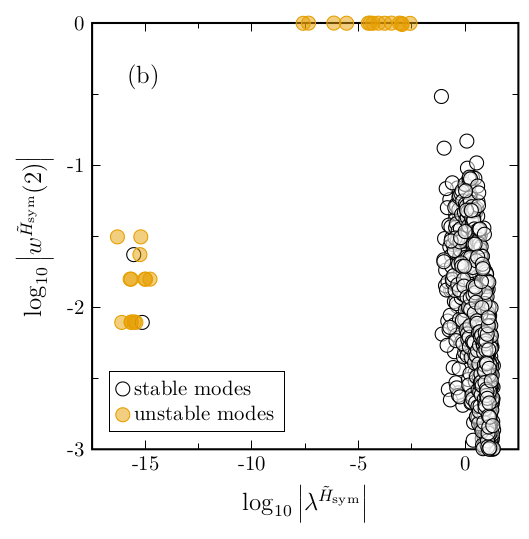}
    \end{minipage}
    \begin{minipage}{0.32\linewidth}
        \centering
        \includegraphics[width=\linewidth]{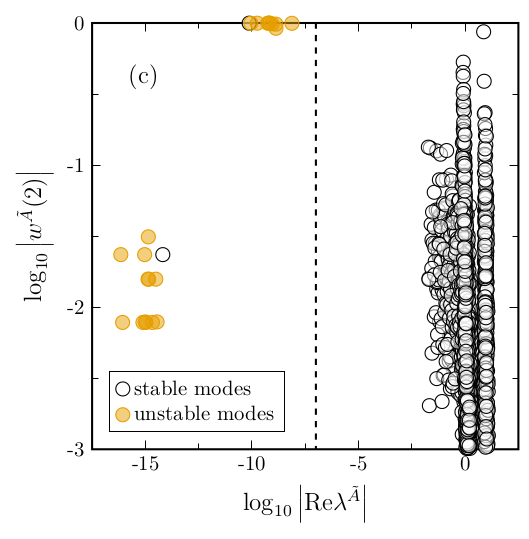}
    \end{minipage}
    \caption{
    Scatter plots of $w^A(2)$ versus $\lambda^A$ for $A=\tilde{\mathcal R}, \tilde H_{\mr{sym}},$ and $\tilde A$ in panels (a), (b), and (c), respectively.
    The $x$-axis in panel (c) shows the real parts of the eigenvalues of $\tilde A$, which are generally complex.
    These results were obtained using 10 samples with $N=128$, $\mu=1.00$, and $\phi=0.600$.
    For better visibility, only one out of every five data points is shown for the stable modes.
    The dashed line in panel (c) indicates the rough estimate of the numerical precision.
    The definitions of stable and unstable modes are different for different matrices; see the main text for details.
    }
    \label{fig:rattler weight}
\end{figure}

Although pseudo-rattlers appear to be stably constrained by the rigid structure, they give rise to unstable potential energy landscapes in the DEM and thus must be removed to obtain stable solid-like structures.
Fig.~\ref{fig:rattler weight} shows scatter plots of $w_i^A(2)$ versus $\lambda_i^A$ for $A = \tilde{\mathcal R}, \tilde H_{\mr{sym}},$ and $\tilde A$ in panels (a), (b), and (c), respectively.
Since the eigenvalues of $\tilde A$ are generally complex, we plot their real parts on the $x$-axis in panel (c).
The dashed line in panel (c) indicates the rough estimate of the numerical precision mentioned in Sec.~\ref{sec: sample preparation}.
The definitions of stable and unstable modes differ for the three matrices:
For $A=\tilde{\mathcal R}$, the stable (unstable) modes satisfy $\lambda^{\tilde{\mathcal R}}>0$ ($\lambda^{\tilde{\mathcal R}}<0$).
For $A=\tilde H_{\mr{sym}}$, the stable (unstable) modes satisfy $\lambda^{\tilde H_{\mr{sym}}}<0$ ($\lambda^{\tilde H_{\mr{sym}}}>0$).
For $A=\tilde A$, the stable (unstable) modes satisfy $\Re\lambda^{\tilde A}<0$ ($\Re\lambda^{\tilde A}>0$).

For the purpose of discussion we divide the data points in Fig.~\ref{fig:rattler weight}(a) into four groups.
The first group consists of modes with eigenvalues close to zero $\lesssim 10^{-10}$.
The second group consists of unstable modes strongly localized on pseudo-rattlers $w_i^{\tilde{\mathcal R}}(2) \sim 1$.
The third and fourth groups consist of unstable and stable modes, respectively, both characterized by relatively small values of $w_i^{\tilde{\mathcal R}}(2)$.
The first group contains zero modes other than the trivial zero modes corresponding to the global translations.  
The second group indicates that the motions of pseudo-rattlers are largely governed by the spurious instability as discussed in the previous section.
The origin of the spurious instability is not limited to pseudo-rattlers because the matrix $\tilde{\mathcal R}$ also possesses unstable modes that are unrelated to pseudo-rattlers and belong to the third group.

Fig.~\ref{fig:rattler weight}(b) shows results similar to those in panel (a) except that the third group is absent, indicating that the unstable modes of $\tilde H_{\mr{sym}}$ originate solely from pseudo-rattlers.
Since $-\tilde H_{\mr{sym}}$ defines the potential energy in $\tilde {\mathcal E}$, its unstable modes correspond to saddle directions at a given inherent structure.
Such saddle modes in the potential energy landscape, however, are typically associated with liquid structures~\cite{Angelani2000Saddles} and considered to be another artifact of the DEM.
In Appendix~\ref{sec: two walls} we show that the motion of a pseudo-rattler can be highly unstable due to both the spurious instability of $\tilde R$ and saddle modes of $-\tilde H_{\mr{sym}}$.

The results in Fig.~\ref{fig:rattler weight}(c) are almost identical to those in panel (b), except that the absolute values of the eigenvalues in the second group $\sim 10^{-8}$ are smaller than in panel (b).
The eigenvalues in the second group lie below the threshold indicated by the dashed line. 
As a result, these unstable modes are regarded as stable numerically and the system therefore appears to have converged to an inherent structure.
This stabilization is attributed to the strong damping effect of the Stokes-type interactions $\Gamma_{\mr{St}}$.

To verify this claim, Fig.~\ref{fig:rattler weight no stokes} shows the same data as Fig.~\ref{fig:rattler weight}, but computed without $\Gamma_{\mr{St}}$.
Since $\Gamma_{\mr{St}}$ does not affect $\tilde H_{\mr{sym}}$, the corresponding results to Fig.~\ref{fig:rattler weight}(b) are not shown.
Fig.~\ref{fig:rattler weight no stokes}(a) shows that the number of unstable modes of $\tilde{\mathcal R}$ increased significantly when $\Gamma_{\mr{St}}$ is removed.  
Fig.~\ref{fig:rattler weight no stokes}(b) shows that the absolute values of the unstable modes of $\tilde A$ are larger than those in Fig.~\ref{fig:rattler weight}(c) and, in particular, the threshold value indicated by the dashed line.
Thus the motions of pseudo-rattlers are essentially unstable in the DEM and only marginally stabilized by the dissipative interactions $\Gamma_{\mr{St}}$, implying that pseudo-rattlers need to be removed in the same way as rattlers\footnote{
The instability of pseudo-rattlers is considered to be responsible for the slow convergence to inherent structures when the dissipative interactions $\Gamma_{\mr{St}}$ are suppressed as mentioned in Sec.~\ref{sec: model and methods}.
}.
In Appendix~\ref{sec: alternative torque}, we describe a simpler alternative definition of the pair torque in Eq.~\eqref{pair torque} and present numerical results demonstrating that the instability associated with pseudo-rattlers remains unchanged under this alternative definition.

As shown in Fig.~\ref{fig:rattlers}, however, each contact of pseudo-rattlers carries a finite force and torque and simply removing such particles invalidates the force and torque balance.
Thus, in our simulation, we removed all pseudo-rattlers, let the system relax until force and torque balance was restored, and repeated this procedure until each particle had more than two contacts.
Although this procedure leaves us with mechanically stable structures suitable for analyzing the spectrum of the energy matrix, this comes at the cost of causing some initially jammed configurations to become unjammed for certain values of $N$, $\mu$, and $\phi$ because the algorithm adopted here is stricter than the standard rattler-removal procedure implemented, for example, in LAMMPS.
See also Table \ref{tab:samples}.
Note that systems with larger friction coefficients tend to have more pseudo-rattlers because larger tangential forces are allowed for larger $\mu$.

\begin{figure}
    \centering

    \begin{minipage}{0.32\linewidth}
        \centering
        \includegraphics[width=\linewidth]{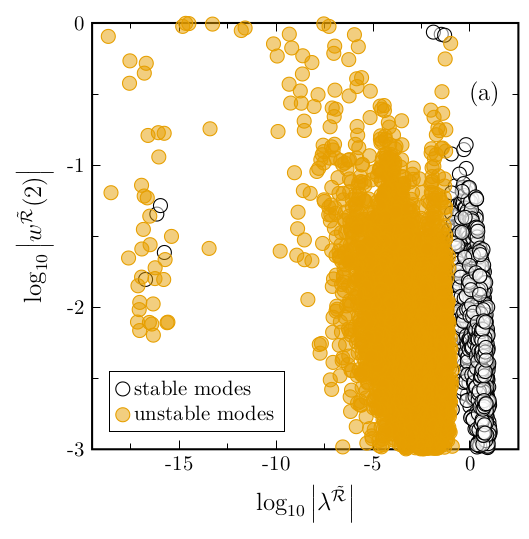}
    \end{minipage}
    \begin{minipage}{0.32\linewidth}
        \centering
        \includegraphics[width=\linewidth]{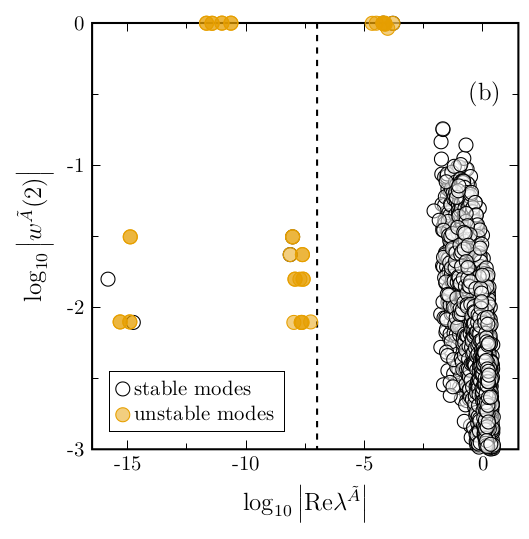}
    \end{minipage}
    \caption{
    Scatter plots of $w^A(2)$ versus $\lambda^A$ for $A=\tilde{\mathcal R}$ and $\tilde A$ in panels (a) and (b), respectively. 
    Panels (a) and (b) correspond to Fig.~\ref{fig:rattler weight}(a) and (c), respectively, except that the contribution of $\Gamma_{\mr{St}}$ is neglected.
    }
    \label{fig:rattler weight no stokes}
\end{figure}

\subsection{Energy matrix}\label{sec: density of states}

Here we investigate the spectrum of the energy matrix $\tilde H_{\mr{sym}}$ after removing rattlers and pseudo-rattlers.
Fig.~\ref{fig:spectra} shows its density of states $D(\omega)$.
We show the results for $N=1024$ and $\mu=0.01, 0.20$, and $1.00$ in panels (a), (b), and (c), respectively, at several values of $\phi$.
In panel (a), the density of states for $\mu=0.00$ and $\phi=0.69$ is also shown for comparison.
The density of states of frictional systems is shifted to the right relative to that of frictionless systems because the additional tangential interactions increase the energy scale.
From Fig.~\ref{fig:spectra}, three frequency regimes can be identified: low-, intermediate-, and high-frequency regimes.
The density of states exhibits a plateau $D(\omega) \sim \mr{const.}$ in the intermediate-frequency regime and this plateau extends to lower frequencies for smaller $\phi$ and larger $\mu$, similarly to frictionless systems.
In the high- and low-frequency regimes, the density of states decays rapidly while the latter appears noisier due to phonon modes as discussed below.
Overall, the qualitative behavior in the intermediate- and high-frequency regimes depends only weakly on $\phi$ and $\mu$.
The similarity between the results for the frictionless and frictional systems indicates that these systems have similar complex moduli and, consequently, exhibit similar viscoelastic responses~\cite{Lemaitre2006Sum,Tighe2011Relaxations,Hara2025Alink,Koyama2025Enhanced}.

To discuss the structure of the eigenmodes, Fig.~\ref{fig:mu_020} shows scatter plots of $W^{(r)}$, $P^{(r)}$, and $P^{(\theta)}$ versus $\omega$ in panels (a), (b), and (c), respectively, for $N=1024$ and $\mu=0.20$.
In the high-frequency regime, $W^{(r)} \ll W^{(\theta)} = 1-W^{(r)}$ holds, meaning that the modes are dominated by rotational degrees of freedom.
From panels (b) and (c), these modes are found to be strongly localized on a few particles.
In the intermediate-frequency regime, the weight $W^{(r)}\sim0.6$ and participation ratios $P^{(r)}, P^{(\theta)}\sim0.5$ are independent of $\omega$.
Both translational and rotational degrees of freedom contribute to the modes almost equally and these modes are relatively extended over the system.
It is noteworthy that the density of states exhibits a plateau in this regime, as in frictionless systems, despite the additional rotational degrees of freedom introduced by friction.
Fig.~\ref{fig:mu_020} demonstrates that these intermediate- and high-frequency properties are qualitatively independent of the packing fraction.
These modes are also qualitatively independent of the friction coefficient as shown in Appendix~\ref{sec: additional numerical results}.
We note that the frequencies of the localized modes dominated by rotational degrees of freedom change depending on the ratio $k_s/k_n$~\cite{Ishima2023Theory}.

In the low-frequency regime, one can identify a few discrete bands of eigenmodes whose weight $W^{(r)}$ is almost one and participation ratio of the translational degrees of freedom $P^{(r)}$ is approximately $0.7$.
These features are characteristic to phonon modes as discussed in previous studies of frictionless systems~\cite{Lerner2021Low-energy}.
In addition, several localized modes coexist with these phonon modes, which is also a characteristic of frictionless systems~\cite{Lerner2021Low-energy}.
As discussed in the previous section, however, the low-frequency data should be treated with caution due to the spurious instability and we cannot determine whether the low-frequency properties of $\tilde H_{\mr{sym}}$ are indeed similar to those of frictionless systems in the case of true Coulomb friction.

\begin{figure}
    \centering

    \begin{minipage}{0.32\linewidth}
        \centering
        \includegraphics[width=\linewidth]{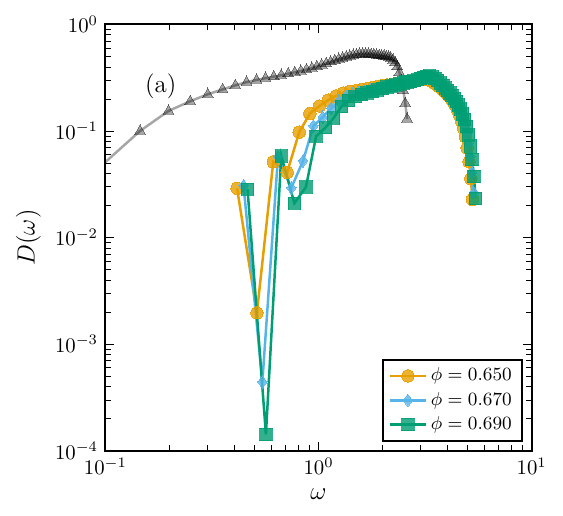}
    \end{minipage}
    \begin{minipage}{0.32\linewidth}
        \centering
        \includegraphics[width=\linewidth]{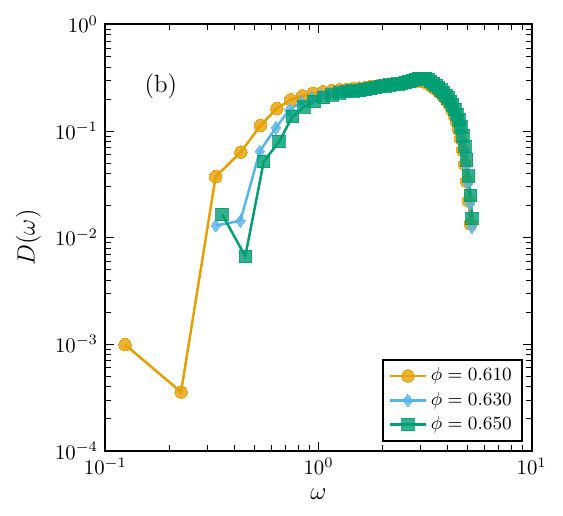}
    \end{minipage}
    \begin{minipage}{0.32\linewidth}
        \centering
        \includegraphics[width=\linewidth]{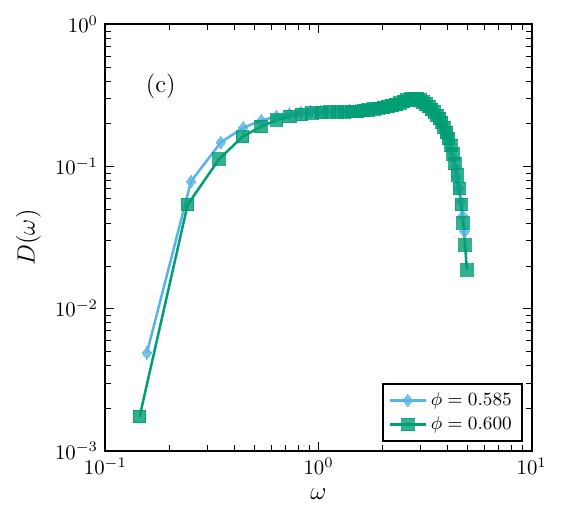}
    \end{minipage}
    \caption{
    Densities of states of the energy matrix $\tilde H_{\mr{sym}}$ for $\mu=0.01, 0.20,$ and $1.00$ in panels (a), (b), and (c), respectively.
    In panel (a), the density of states for $\mu=0.00$ and $\phi=0.69$ is also shown as black triangles for comparison.
    These results were obtained using samples with $N=1024$.
    }
    \label{fig:spectra}
\end{figure}

\section{Discussion}\label{sec: discussion}

As discussed in Secs.~\ref{sec: theory} and \ref{sec: numerical results}, the DEM suffers from several artifacts that hinder a detailed investigation of the linear viscoelastic response, particularly in the low-frequency regime. 
Here, we briefly discuss three possible alternative approaches to modeling Coulomb friction to overcome this difficulty and describe the advantages and disadvantages of each. 
First, one can consider an effective potential instead of modeling contact friction with history-dependent interactions~\cite{Liu2021Spongelike}. 
This approach is quite convenient because it allows us to utilize methods developed for frictionless systems. 
However, by construction, it cannot describe phenomena that intrinsically rely on the history dependence of contact friction.

Second, one can model contact friction by introducing internal degrees of freedom into each particle. 
One simple example is to introduce a number of small patches on the surface of each particle, and patches on different particles attract each other when they are sufficiently close.
Such models have already been proposed in the context of colloidal systems~\cite{Immink2019Using}.
However, realistically modeling Coulomb friction would require a large number of patches, making large-scale simulations computationally demanding.

Finally, one can completely abandon smooth interactions and instead adopt the nonsmooth contact dynamics method~\cite{Jean1999The}, which treats Coulomb friction more rigorously than the DEM by solving the resulting system of inequalities. 
Because this framework enforces the exact Coulomb friction conditions, it can naturally account for the history-dependent nature of contact friction without introducing any artifacts. 
However, because the interactions in this framework are inherently nondifferentiable, linear stability analysis based on the Jacobian matrix cannot be directly applied.

\section{Summary and conclusion}\label{sec: summary}

In this work, we investigated the linear stability of frictional granular systems modeled by the DEM.
By linearizing the equations of motion, we derived the energy and dissipation matrices from the Jacobian matrix of the forces and torques, providing a framework for analyzing the linear response of the granular system.
In particular, these matrices control the viscoelastic response of the system and are therefore relevant to rheological measurements.
By analyzing these matrices, we identified two artifacts of the DEM that hinder a physically meaningful characterization of the low-frequency dynamics.

First, the dissipation matrix inevitably possesses negative modes, indicating spurious energy injection through such unstable modes. 
This behavior is unique to externally driven systems, and is clearly inconsistent with the interpretation of the DEM as a model of Coulomb friction, which should only dissipate energy. 
Second, particles with only two contacts, referred to here as pseudo-rattlers, generate negative modes in the energy matrix, implying that the system is located at a saddle point of the potential energy landscape. 
These artifactual instabilities are only marginally stabilized in the full Jacobian matrix by Stokes-type dissipative interactions.

Taking these artifacts into account, we examined the eigenmodes of the energy matrix. 
Analysis of several fundamental observables, including the density of states, revealed that the intermediate- and high-frequency regimes are qualitatively insensitive to both the friction coefficient and the packing fraction. 
In particular, the intermediate-frequency regime exhibits a plateau reminiscent of that observed in frictionless granular systems~\cite{Silbert2005Vibrations}, despite the contribution of rotational degrees of freedom.
The low-frequency regime likewise shares several features with frictionless systems, including discretized phonon bands and localized modes. 
However, since the influence of the aforementioned artifacts on this regime remains unclear, these observations should be interpreted with caution.

Finally, we discussed several alternative simulation frameworks that do not rely on virtual tangential springs to model contact friction.
At present, however, none of them appears to be fully capable of describing the linear viscoelastic response of frictional granular systems.
Developing theoretical and computational approaches that overcome the limitations of existing frameworks remains an important challenge for future research of frictional granular systems.

\begin{figure}
    \centering

    \begin{minipage}{0.32\linewidth}
        \centering
        \includegraphics[width=\linewidth]{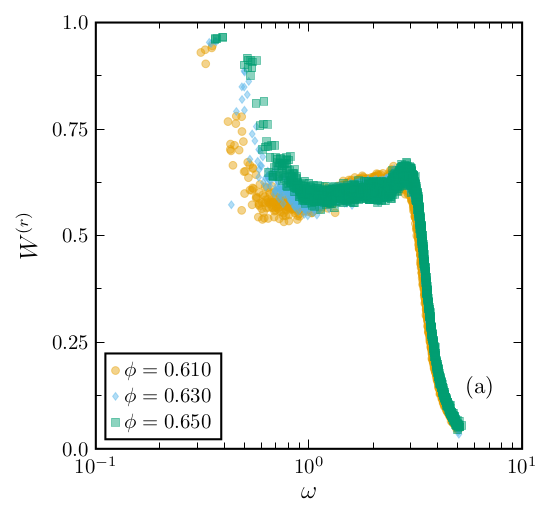}
    \end{minipage}
    \begin{minipage}{0.32\linewidth}
        \centering
        \includegraphics[width=\linewidth]{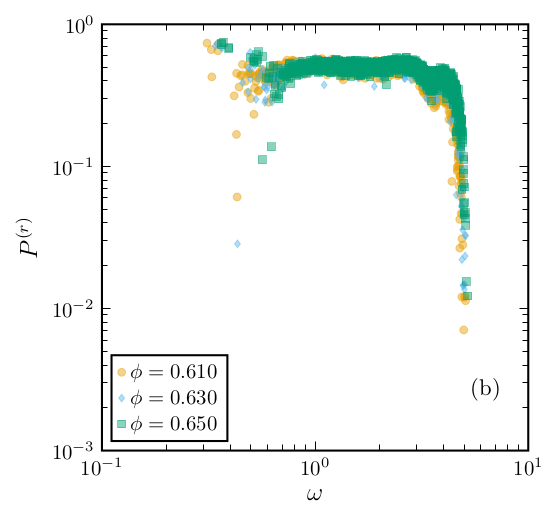}
    \end{minipage}
    \begin{minipage}{0.32\linewidth}
        \centering
        \includegraphics[width=\linewidth]{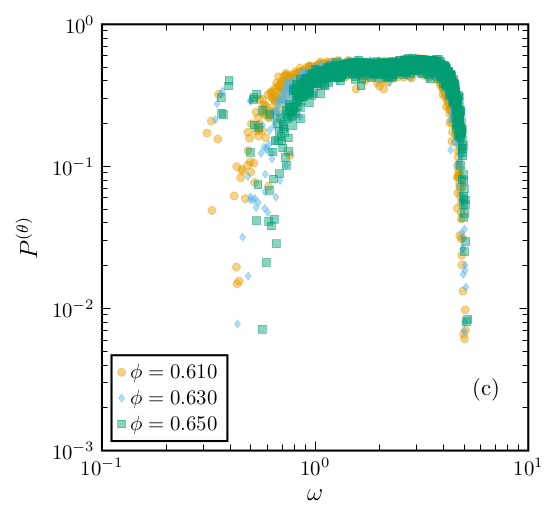}
    \end{minipage}
    \caption{
    Scatter plots of $W^{(r)}$, $P^{(r)}$, and $P^{(\theta)}$ versus $\omega$ in panels (a), (b), and (c), respectively.
    These results were obtained using 10 samples with $N=1024$ and $\mu=0.20$.
    }
    \label{fig:mu_020}
\end{figure}

\section*{Conflict of interest}

The authors declare no conflicts of interest.

\section*{Acknowledgment}

We thank Norihiro Oyama, Atsushi Ikeda, and Yuji Ito for their insightful comments.
In this research work, we used the ``mdx: a platform for building data-empowered society''~\cite{Suzumura2022mdx}.

\section*{Declaration of generative AI and AI-assisted technologies in the manuscript preparation process}

During the preparation of this work, the authors used Open AI ChatGPT (GPT-5) in order to assist with code implementation and to improve the language and clarity of the manuscript.
After using this tool/service, the authors reviewed and edited the content as needed and take full responsibility for the content of the published article.

\appendix

\section{Mathematical preliminaries}\label{sec: mathematical preliminaries}

\subsection{Notation}\label{sec: notation}

Let $I_D$ denote the $D$-dimensional identity matrix, and let $O_{D\times D'}$ denote the $D\times D'$ zero matrix.
In particular $O_D \coloneqq O_{D\times D}$ denotes the $D$-dimensional square zero matrix.
The symmetric and skew-symmetric parts of a square matrix $A$ are defined as  
\eq{
    A_{\mr{sym}} \coloneqq \frac{A + A^T}{2} , \\
    A_{\mr{skew}} \coloneqq \frac{A - A^T}{2} .
}
In the main text, we often define a three-dimensional vector associated with particle $i$ such as the position vector $\bs{r}_i = ( x_i\,\,\, y_i\,\,\, z_i )^T$, where $\bullet^T$ denotes the transpose, and collecting $N$ such three-dimensional vectors we write
\eq{
    \bs{r} & \coloneqq 
    \begin{pmatrix}
        \bs{r}_1^T & \cdots & \bs{r}_N^T
    \end{pmatrix}^T .
}
Tensor contractions are explicitly indicated by dots.
For example, given a vector $\bs{x}$ and matrices $A$ and $B$, we write
\eq{
    [A \cdot \bs{x}]_\mu & = \sum_\nu A_{\mu\nu} x_\nu , \\
    [\bs{x} \cdot A]_\mu & = \sum_\nu x_\nu A_{\nu\mu} , \\
    [A \cdot B]_{\mu\nu} & = \sum_\lambda A_{\mu\lambda} B_{\lambda\nu} ,
}
where $\mu$ and $\nu$ are indices of the tensors and $\mu$ should not be confused with the friction coefficient.
The $i$-th eigenvalue and eigenvector of a matrix $A$ are written as
\eq{
    A \cdot \bs{e}_i^A = \lambda^A_i \bs{e}^A_i ,
}
where the eigenvectors are normalized $|\bs{e}_i^A|=1$.

All tensors in the remainder of this subsection are three-dimensional.
The cross product $\times$ is defined as contractions of two tensors and the Levi-Civita symbol $\varepsilon$
\eq{
    [ \bs{x}_1 \times \bs{x}_2 ]_\mu & = \sum_{\nu\lambda} \varepsilon_{\mu\nu\lambda} x_{1,\nu} x_{2,\lambda} , \label{p1p2} \\
    [ A \times \bs{x} ]_{\mu\nu} & = \sum_{\lambda\rho} A_{\mu\lambda} \varepsilon_{\nu\lambda\rho} x_{\rho} , \\
    [ \bs{x} \times A ]_{\mu\nu} & = \sum_{\lambda\rho} \varepsilon_{\mu\lambda\rho} x_\lambda A_{\rho\nu} .
}
A skew-symmetric matrix $\bs{x}^\times$ associated with a vector $\bs{x}$ is defined as
\eq{
    [\bs{x}^\times]_{\mu\nu} & \coloneqq - \sum_\lambda \varepsilon_{\mu\nu\lambda} x_\lambda ,
}
which is explicitly written as
\eq{
    \bs{x}^\times & = 
    \begin{pmatrix}
        0 & -x_3 & x_2 \\
        x_3 & 0 & -x_1 \\
        -x_2 & x_1 & 0 
    \end{pmatrix}
    .
}
Using this expression, we can write
\eq{
    \bs{x}_1 \times \bs{x}_2 & = \bs{x}_1^\times \cdot \bs{x}_2 = - \bs{x}_2^\times \cdot \bs{x}_1 , \\
    A \times \bs{x} & = A \cdot \bs{x}^\times , \\
    \bs{x} \times A & = \bs{x}^\times \cdot A .
}
In particular, for the three-dimensional identity matrix $I_3$,
\eq{
    I_3 \times \bs{x} & = \bs{x} \times I_3 = \bs{x}^\times , \\
    \hat{\bs{x}} \times I_3 \times \hat{\bs{x}} & = \hat{\bs{x}}^\times \cdot \hat{\bs{x}}^\times = - I_3 + \mathcal P^\parallel (\bs{x}) = - \mathcal P^\perp (\bs{x}) ,
}
where $\hat{\bs{x}} \coloneqq \bs{x}/|\bs{x}|$ is the unit vector in the direction of $\bs{x}$ and the projection operators are given by
\eq{
    \mathcal P^\parallel (\bs{x}) & \coloneqq \hat{\bs{x}}\hat{\bs{x}} , \label{P parallel} \\
    \mathcal P^\perp (\bs{x}) & \coloneqq I_3 - \mathcal P^\parallel (\bs{x}) .
}
A combination of tensors written without any symbol, e.g., $\hat{\bs{x}}\hat{\bs{x}}$ in Eq.~\eqref{P parallel}, denotes the tensor product.
We often write the magnitude of a vector using the corresponding non-bold symbol, $x \coloneqq |\bs{x}|$, when no confusion arises.

\subsection{Technical lemmas}\label{sec: technical lemmas}

The following lemma is a special case of the result on page 651 of Ref.~\cite{Boyd2004Convex}.

\begin{lemma}
Let $A$ be a real block matrix of the form
\eq{
    A = 
    \begin{pmatrix}
        O_n & B \\
        B^T & C 
    \end{pmatrix}
    , \label{lemma}
}
where $B$ is an $n\times m$ real matrix and $C$ is an $m\times m$ real symmetric matrix.
Then $A$ is positive semidefinite if and only if $B = O_{n\times m}$ and $C$ is positive semidefinite.
\end{lemma}

\begin{proof}
Suppose that $A$ is positive semidefinite. 
Assume to the contrary that $B\neq O_{n\times m}$.
Then there exists a vector $\bs{x}\in\mathbb{R}^m$ such that $B \cdot \bs{x}\neq 0$.
For $t\in\mathbb{R}$, define
\eq{
    \bs{w}_t \coloneqq 
    \begin{pmatrix} 
        tB \cdot \bs{x} \\ 
        \bs{x} 
    \end{pmatrix}.
}
Then
\eq{
    \bs{w}_t \cdot A \cdot \bs{w}_t = 2t | B \cdot \bs{x} |^2 + \bs{x} \cdot C \cdot \bs{x}.
}
Since $B\cdot \bs{x}\neq 0$, the right-hand side becomes negative for sufficiently large negative $t$, contradicting the positive semidefiniteness of $A$. 
Hence $B=O_{n\times m}$.
It then follows immediately that $C$ is positive semidefinite.

Conversely, if $B=O_{n\times m}$ and $C$ is positive semidefinite, $A$ is obviously positive semidefinite.
\end{proof}
The negative semidefinite case follows by reversing the signs and inequalities and thus we obtain the following result:

\begin{corollary}
If $B \neq O_{n\times m}$ in Eq.~\eqref{lemma}, then $A$ is indefinite.
\end{corollary}

\section{Tangential force}\label{sec: tangential force}

Suppose that the contact between particles $i$ and $j$ is formed at $t=t_0$ and lost at $t=t_f$.
The tangential displacement $\bs{\xi}_{ij}$ is zero at $t=t_0$, updated based on the algorithm explained here, and reset to zero at $t=t_f$.
During the contact, the vectors $\bs{r}_{ij}$ and $\bs{\xi}_{ij}$ are required to remain perpendicular at all times
\eq{
    \bs{r}_{ij} (t) \cdot \bs{\xi}_{ij} (t) & = 0 . \label{orthogonality}
}
Using the tangential displacement $\bs{\xi}_{ij}$ at time $t$, where $t_0\leq t < t_f$, the time derivative of the bare displacement $\bs{\xi}_{ij}^{\mr{bare}}$ is simply proportional to the tangential relative velocity
\eq{
    \dv{\bs{\xi}_{ij}^{\mr{bare}} (t)}{t} & = \bs{v}_{ij}^\perp (t) \coloneqq \mathcal P^\perp (\bs{r}_{ij}(t)) \cdot \bs{v}_{ij} (t) - \bs{\nu}_{ij} (t) \times \hat{\bs{r}}_{ij} (t) . \label{xi bare}
}
This equation of motion is solved up to time $t + \dd t$ with an initial condition $\bs{\xi}^{\mr{bare}}_{ij} (t) = \bs{\xi}_{ij}(t)$.
Since $\bs{\xi}^{\mr{bare}}_{ij}(t+\dd t)$ is not necessarily orthogonal to $\hat{\bs{r}}_{ij}(t + \dd t)$ as required in Eq.~\eqref{orthogonality} due to a rigid-body rotation of the pair $ij$, the bare displacement at $t + \dd t$ needs to be rotated as\footnote{
Strictly, there are two modes of rigid-body rotation: ``tumbling'' and ``spinning'' motions.
In many numerical algorithms including LAMMPS~\cite{Plimpton1995Fast,Thompson2022LAMMPS}, however, only the former is corrected as in Eq.~\eqref{sticking disp}.
See the section of \texttt{pair\_style granular} in the LAMMPS documentation.
This is not fatal because the latter mode does not rotate the contact normal at a contact point.
}
\eq{
    \bs{\xi}_{ij}^{\mr{stick}} (t + \dd t) & \coloneqq \mathcal P^\perp (\bs{r}_{ij}(t + \dd t)) \cdot \bs{\xi}^{\mr{bare}}_{ij} (t + \dd t) \frac{|\bs{\xi}^{\mr{bare}}_{ij}(t+\dd t)|}{|\mathcal P^\perp (\bs{r}_{ij}(t + \dd t)) \cdot \bs{\xi}^{\mr{bare}}_{ij} (t + \dd t)|} , \label{sticking disp}
}
which we call the sticking displacement.
Expanding $\mathcal P^\perp (\bs{r}_{ij}(t + \dd t))$ and $\bs{\xi}_{ij}^{\mr{bare}} (t + \dd t)$ in Eq.~\eqref{sticking disp}, we obtain
\eq{
    \mathcal P^\perp (\bs{r}_{ij}(t + \dd t)) \cdot \bs{\xi}^{\mr{bare}}_{ij} (t + \dd t) & = \left ( I_3 - \frac{1}{r_{ij}(t)^2} \qty ( \bs{r}_{ij}(t) \bs{r}_{ij}(t) + \dd \bs{r}_{ij}(t) \bs{r}_{ij}(t) + \bs{r}_{ij}(t) \dd \bs{r}_{ij}(t) ) \right. \notag \\
    & \qquad \left. - \frac{1}{r_{ij}(t)^4} \qty ( \bs{r}_{ij}(t) \bs{r}_{ij}(t) ) \qty ( - 2 \bs{r}_{ij} (t) \cdot \dd \bs{r}_{ij}(t) ) \right ) \cdot \qty ( \bs{\xi}_{ij} (t) + \dd \bs{\xi}^{\mr{bare}}_{ij} (t) ) + o (\dd t) \notag \\
    & = \bs{\xi}_{ij} (t) - \frac{\hat{\bs{r}}_{ij}(t) \bs{\xi}_{ij} (t)}{r_{ij}(t)} \cdot \dd \bs{r}_{ij}(t) + \dd \bs{\xi}^{\mr{bare}}_{ij} (t) + o (\dd t) , \label{rotate 1}
}
which leads to
\eq{
    |\mathcal P^\perp (\bs{r}_{ij}(t + \dd t)) \cdot \bs{\xi}^{\mr{bare}}_{ij} (t + \dd t)| & = \sqrt{\bs{\xi}_{ij}(t)^2 + 2 \bs{\xi}_{ij}(t) \cdot \dd \bs{\xi}^{\mr{bare}}_{ij}(t)} + o (\dd t) = \qty | \bs{\xi}^{\mr{bare}}_{ij} (t + \dd t) | + o (\dd t) . \label{rotate 2}
}
Eqs.~\eqref{sticking disp}--\eqref{rotate 2} yield the equation of motion for the sticking displacement
\eq{
    \dv{\bs{\xi}_{ij}^{\mr{stick}} (t)}{t} \coloneqq \lim_{\dd t\to0} \frac{\bs{\xi}_{ij}^{\mr{stick}} (t + \dd t) - \bs{\xi}_{ij} (t)}{\dd t} & = \bs{v}_{ij}^\perp (t) - \frac{\hat{\bs{r}}_{ij}(t)\bs{\xi}_{ij}(t)}{r_{ij}(t)} \cdot \bs{v}_{ij} (t) ,
}
where the term ${\hat{\bs{r}}_{ij}\bs{\xi}_{ij}} \cdot \bs{v}_{ij} /{r_{ij}}$ can be interpreted as an inertial force that enforces the tangential displacement to remain perpendicular to $\bs{r}_{ij}$~\cite{Silbert2001Granular}.
With the sticking displacement $\bs{\xi}_{ij}^{\mr{stick}}$, the ``test force'' at time $t+\dd t$ is defined as
\eq{
    \bs{G}_{ij} (t+\dd t) & \coloneqq k_{s,ij} ( \Delta_{ij} (t + \dd t) ) \bs{\xi}_{ij}^{\mr{stick}} (t+\dd t) + \eta_{s,ij} (\Delta_{ij}(t+\dd t)) \bs{v}_{ij}^\perp (t + \dd t) .
}
The test force is selected as the tangential force $\bs{F}_{ij}^\perp = \bs{G}_{ij}$ at time $t + \dd t$ if its magnitude is below the Coulomb limit
\eq{
    |\bs{G}_{ij} (t + \dd t)| < \mu |\bs{F}_{ij}^\parallel (t + \dd t)| . \label{sticking}
}
In this case the contact between particles $i$ and $j$ is considered to be sticking due to static friction.
The tangential displacement at $t+\dd t$ is given by $\bs{\xi}_{ij} (t+ \dd t) = \bs{\xi}_{ij}^{\mr{stick}} (t + \dd t)$.
If Eq.~\eqref{sticking} does not hold, the tangential displacement is rescaled so that the resulting tangential force satisfies $|\bs{F}_{ij}^{\perp} (t + \dd t)| = \mu |\bs{F}_{ij}^\parallel (t + \dd t)|$:
\eq{
    \bs{\xi}_{ij}^{\mr{slip}} (t + \dd t) & \coloneqq \frac{1}{k_{s,ij} (t + \dd t)} \qty ( \mu |\bs{F}_{ij}^\parallel (t + \dd t)| \hat{\bs{G}}_{ij} (t + \dd t) - \eta_{s,ij} ( \Delta_{ij} ( t + \dd t) ) \bs{v}_{ij}^\perp (t + \dd t) ) ,
}
which leads to
\eq{
    \bs{F}_{ij}^\perp (t + \dd t) & = k_{s,ij} ( \Delta_{ij} ( t + \dd t) ) \bs{\xi}_{ij}^{\mr{slip}} (t + \dd t) + \eta_{s,ij} (\Delta_{ij} (t + \dd t)) \bs{v}_{ij}^\perp (t + \dd t) = \mu |\bs{F}_{ij}^\parallel (t + \dd t)| \hat{\bs{G}}_{ij} (t + \dd t) .
}
Then the tangential displacement at time $t + \dd t$ is given by $\bs{\xi}_{ij} (t + \dd t) = \bs{\xi}_{ij}^{\mr{slip}} (t + \dd t)$ and the contact is considered to be slipping as a result of the transition from static to kinetic friction.

In summary, starting from the tangential displacement $\bs{\xi}_{ij}$ at time $t$, the update rule is as follows.
\begin{enumerate}
    \item Solve the equation of motion of the sticking displacement up to time $t+\dd t$
    \eq{
        \dv{\bs{\xi}_{ij}^{\mr{stick}}}{t} & = \bs{v}_{ij}^\perp - \frac{\hat{\bs{r}}_{ij}\bs{\xi}^{\mr{stick}}_{ij}}{r_{ij}} \cdot \bs{v}_{ij} \label{ODE for xi}   
    }
    with an initial condition $\bs{\xi}_{ij}^{\mr{stick}}(t) = \bs{\xi}_{ij}(t)$.
    \item Compute the test force
    \eq{
        \bs{G}_{ij} (t+\dd t) = k_{s,ij}(\Delta_{ij}(t + \dd t)) \bs{\xi}_{ij}^{\mr{stick}} (t+\dd t) + \eta_{s,ij} (\Delta_{ij}(t + \dd t)) \bs{v}_{ij}^\perp (t + \dd t) .
    }
    \item If the Coulomb condition is satisfied, $|\bs{G}_{ij}(t + \dd t)| < \mu |\bs{F}_{ij}^\parallel (t + \dd t)|$, set 
    \eq{
        \bs{F}_{ij}^\perp (t + \dd t) & = \bs{G}_{ij} (t + \dd t) , \\
        \bs{\xi}_{ij}(t + \dd t) & = \bs{\xi}_{ij}^{\mr{stick}}(t + \dd t) .
    }
    Otherwise, set
    \eq{
        \bs{F}_{ij}^\perp (t + \dd t) & = \mu | \bs{F}_{ij}^\parallel (t + \dd t) | \hat{\bs{G}}_{ij} (t + \dd t) , \label{F slip} \\
        \bs{\xi}_{ij} (t + \dd t) & = \bs{\xi}_{ij}^{\mr{slip}} (t + \dd t) = \frac{1}{k_{s,ij}(\Delta_{ij}(t + \dd t))} \qty ( \mu | \bs{F}_{ij}^\parallel (t + \dd t) | \hat{\bs{G}}_{ij} (t + \dd t) - \eta_s (\Delta_{ij}(t + \dd t)) \bs{v}_{ij}^\perp(t+\dd t) ) \label{xi slip} .
    }
    \item Update the time $t \to t + \dd t$ and return to step 1.
\end{enumerate}
Eq.~\eqref{F perp} schematically represents this procedure.
Eqs.~\eqref{F slip} and \eqref{xi slip} give rise to nondifferentiability in this model.

From the definition of $\bs{F}_{ij}^\perp$, the following inequality holds at all times:
\eq{
    |\bs{F}_{ij}^\perp| \leq \mu |\bs{F}_{ij}^\parallel| . \label{Coulomb}
}
Since we only consider infinitesimal perturbations of static packings, we assume that all contacts are sticking, or equivalently~\cite{Ishima2023Theory}
\eq{
    |\bs{F}_{ij}^\perp| < \mu |\bs{F}_{ij}^\parallel| . \label{assumption on friction}
}
This assumption is natural as long as we consider infinitesimal perturbations. 
It is practically quite difficult, if not impossible, to determine whether the inequality in Eq.~\eqref{Coulomb} is exactly saturated in numerical simulations.
The assumption in Eq.~\eqref{assumption on friction} removes the nondifferentiability from the model and allows us to compute derivatives of the pair force and torque.
The time evolution of the tangential displacement is therefore governed solely by Eq.~\eqref{ODE for xi}:
\eq{
    \dv{\bs{\xi}_{ij}}{t} & = \bs{v}_{ij}^\perp - \frac{\hat{\bs{r}}_{ij}\bs{\xi}_{ij}}{r_{ij}} \cdot \bs{v}_{ij} \label{increment}
}
with an initial condition $\bs{\xi}_{ij}(t_0) = 0$.

\section{Jacobian matrix}\label{sec: derivatives of force and torque}

\subsection{Orientation of a sphere}\label{sec: orientation}

Here we explain the complete definition of the variable $\bs{\theta}$ introduced in Sec.~\ref{sec: linearization}.
The orientation of an object is described by a rotation matrix $\Theta$, which is an element of the special orthogonal group $\mr{SO}(3)$ and written as
\eq{
    \Theta & = \exp \qty ( - \bs{\theta}^\times  ) , \label{SO3}
}
where $\bs{\theta}$ is a three-dimensional vector.
The skew-symmetric matrix $\bs{\theta}^\times$ is an element of the Lie algebra $\mathfrak{so}(3)$, and the vector $\bs{\theta}$ represents the rotation axis by its direction and the rotation angle by its magnitude.
Using the angular velocity $\bs{\omega}$, the time evolution of the rotation matrix is given by
\eq{
    \dv{\Theta}{t} & = - \bs{\omega}^\times \cdot \Theta . \label{equation of motion for R}
}
Substituting Eq.~\eqref{SO3} and assuming $|\bs{\theta}| \ll 1$, we have
\eq{
    \dv{\bs{\theta}^\times}{t} & = \bs{\omega}^\times + \order{\bs{\theta}} .
}
Thus if the orientation changes only infinitesimally, its time evolution is governed by a simple equation
\eq{
    \dv{\bs{\theta}}{t} & = \bs{\omega} , \label{dthetadt}
}
which has the same structure as a two-dimensional rotation.

\subsection{Block matrices}

The matrices $H$ and $\Gamma$ are composed of three-dimensional blocks of the form $\pdv*{\bs{A}_i}{\bs{a}_j}$, where $\bs{A} = \bs{F}, \bs{T}$ and $\bs{a} = \bs{r}, \bs{\theta}, \bs{p}, \bs{L}$.
Since the pair force and torque depend only on the relative quantities, $\bs{r}_{ij}, \bs{\chi}_{ij}, \bs{v}_{ij}, \bs{\nu}_{ij}$, these three-dimensional blocks are written as
\eq{
    \pdv{\bs{A}_i}{\bs{r}_j} & = 
    \begin{dcases}
        \pdv{\bs{A}_{ij}}{\bs{r}_{ij}} & (i\neq j) \\
        - \sum_{\substack{k=1\\(k\neq i)}}^N \pdv{\bs{A}_{ik}}{\bs{r}_{ik}} & (i = j)
    \end{dcases}
    , \label{ar} \\
    \pdv{\bs{A}_i}{\bs{\theta}_j} & = 
    \begin{dcases}
        \pdv{\bs{A}_{ij}}{\bs{\chi}_{ij}} R_j & ( i \neq j ) \\
        \sum_{\substack{k=1\\(k\neq i)}}^N \pdv{\bs{A}_{ik}}{\bs{\chi}_{ik}} R_j & ( i = j ) 
    \end{dcases}
    , \\
    \pdv{\bs{A}_i}{\bs{p}_j} & = 
    \begin{dcases}
        \pdv{\bs{A}_{ij}}{\bs{v}_{ij}} \frac{1}{m_j} & (i\neq j) \\
        -\sum_{\substack{k=1\\(k\neq i)}}^N \pdv{\bs{A}_{ik}}{\bs{v}_{ik}} \frac{1}{m_j} & (i = j)
    \end{dcases} 
    , \label{dAdp} \\
    \pdv{\bs{A}_i}{\bs{L}_j} & = 
    \begin{dcases}
        \pdv{\bs{A}_{ij}}{\bs{\nu}_{ij}} \frac{R_j}{J_j} & ( i \neq j ) \\
        \sum_{\substack{k=1\\(k\neq i)}}^N \pdv{\bs{A}_{ik}}{\bs{\nu}_{ik}} \frac{R_j}{J_j} & ( i = j ) 
    \end{dcases}
    \label{aL}. 
}
Changing the state variable from $\bs{q}$ to $\bs{x}$, we write
\eq{
    \tilde H & =
    \begin{pmatrix}
        \widetilde{\pdv{\bs{F}}{\bs{r}}} & \widetilde{\pdv{\bs{F}}{\bs{\theta}}} \\[0.5em]
        \widetilde{\pdv{\bs{T}}{\bs{r}}} & \widetilde{\pdv{\bs{T}}{\bs{\theta}}} 
    \end{pmatrix} , \\
    \tilde \Gamma 
    & =
    \begin{pmatrix}
        \widetilde{\pdv{\bs{F}}{\bs{p}}} & \widetilde{\pdv{\bs{F}}{\bs{L}}} \\[0.5em]
        \widetilde{\pdv{\bs{T}}{\bs{p}}} & \widetilde{\pdv{\bs{T}}{\bs{L}}} 
    \end{pmatrix} 
    .
}
and Eqs.~\eqref{ar}--\eqref{aL} are rewritten as
\eq{
    \widetilde{\pdv{\bs{F}_i}{\bs{r}_j}} & \coloneqq 
    \begin{dcases}
        \frac{1}{\sqrt{m_im_j}} \pdv{\bs{F}_{ij}}{\bs{r}_{ij}} & (i \neq j) \\
        - \frac{1}{m_i} \sum_{\substack{k=1\\(k\neq i)}}^N \pdv{\bs{F}_{ik}}{\bs{r}_{ik}} & (i = j)
    \end{dcases}
    , \label{tfr} \\
    \widetilde{\pdv{\bs{F}_i}{\bs{\theta}_j}} & \coloneqq 
    \begin{dcases}
        \frac{\zeta}{\sqrt{m_im_j}}\pdv{\bs{F}_{ij}}{\bs{\chi}_{ij}} & ( i \neq j ) \\
        \frac{\zeta}{m_i}\sum_{\substack{k=1\\(k\neq i)}}^N \pdv{\bs{F}_{ik}}{\bs{\chi}_{ik}} & ( i = j ) 
    \end{dcases} , \\
    \widetilde{\pdv{\bs{F}_i}{\bs{p}_j}} & \coloneqq 
    \begin{dcases}
        \frac{1}{\sqrt{m_im_j}} \pdv{\bs{F}_{ij}}{\bs{v}_{ij}} & (i\neq j) \\
        -\frac{1}{m_i} \sum_{\substack{k=1\\(k\neq i)}}^N \pdv{\bs{F}_{ik}}{\bs{v}_{ik}} & (i = j)
    \end{dcases} , \\
    \widetilde{\pdv{\bs{F}_i}{\bs{L}_j}} & \coloneqq 
    \begin{dcases}
        \frac{\zeta}{\sqrt{m_im_j}} \pdv{\bs{F}_{ij}}{\bs{\nu}_{ij}} & ( i \neq j ) \\
        \frac{\zeta}{m_i} \sum_{\substack{k=1\\(k\neq i)}}^N \pdv{\bs{F}_{ik}}{\bs{\nu}_{ik}} & ( i = j ) 
    \end{dcases} , \\
    \widetilde{\pdv{\bs{T}_i}{\bs{r}_j}} & \coloneqq 
    \begin{dcases}
        \frac{\zeta}{\sqrt{m_im_j}} \pdv{\tilde{\bs{T}}_{ij}}{\bs{r}_{ij}} & (i\neq j) \\
        - \frac{\zeta}{m_i} \sum_{\substack{k=1\\(k\neq i)}}^N \pdv{\tilde{\bs{T}}_{ik}}{\bs{r}_{ik}} & (i = j)
    \end{dcases} 
    , \\
    \widetilde{\pdv{\bs{T}_i}{\bs{\theta}_j}} & \coloneqq 
    \begin{dcases}
        \frac{\zeta^2}{\sqrt{m_im_j}} \pdv{\tilde{\bs{T}}_{ij}}{\bs{\chi}_{ij}} & ( i \neq j ) \\
        \frac{\zeta^2}{m_i} \sum_{\substack{k=1\\(k\neq i)}}^N \pdv{\tilde{\bs{T}}_{ik}}{\bs{\chi}_{ik}} & ( i = j ) 
    \end{dcases} 
    , \\
    \widetilde{\pdv{\bs{T}_i}{\bs{p}_j}} & \coloneqq 
    \begin{dcases}
        \frac{\zeta}{\sqrt{m_im_j}} \pdv{\tilde{\bs{T}}_{ij}}{\bs{v}_{ij}} & (i\neq j) \\
        - \frac{\zeta}{m_i} \sum_{\substack{k=1\\(k\neq i)}}^N \pdv{\tilde{\bs{T}}_{ik}}{\bs{v}_{ik}} & (i = j)
    \end{dcases} 
    ,  \\
    \widetilde{\pdv{\bs{T}_i}{\bs{L}_j}} & \coloneqq 
    \begin{dcases}
        \frac{\zeta^2}{\sqrt{m_im_j}} \pdv{\tilde{\bs{T}}_{ij}}{\bs{\nu}_{ij}} & ( i \neq j ) \\
        \frac{\zeta^2}{m_i} \sum_{\substack{k=1\\(k\neq i)}}^N \pdv{\tilde{\bs{T}}_{ik}}{\bs{\nu}_{ik}} & ( i = j ) 
    \end{dcases}
    \label{ttl} ,
}
where $\zeta \coloneqq \sqrt{5/2}$, $\tilde{\bs{T}}_{ij} \coloneqq \bs{T}_{ij}/R_i =  \tilde{\Delta}_{ij} \hat{\bs{r}}_{ij} \times \bs{F}^\perp_{ij}$, and $\tilde \Delta_{ij} \coloneqq 1 - {\Delta_{ij}}/{2R_i}$.

\subsection{First derivatives}

We compute the first derivatives of the pair force and torque in the linear spring-dashpot model with respect to the relative quantities $\bs{r}_{ij}$, $\bs{\chi}_{ij}$, $\bs{v}_{ij}$, and $\bs{\nu}_{ij}$, which appear in Eqs.~\eqref{tfr}--\eqref{ttl}.
The first derivative of the normal force $\bs{F}^\parallel_{ij}$ with respect to $\bs{r}_{ij}$ is given by
\eq{
    \pdv{\bs{F}_{ij}^\parallel}{\bs{r}_{ij}} & = - \pdv{f_{n,ij}(\Delta_{ij})}{\bs{r}_{ij}} \hat{\bs{r}}_{ij} - f_{n,ij} (\Delta_{ij}) \pdv{\hat{\bs{r}}_{ij}}{\bs{r}_{ij}} + \pdv{\qty ( \eta_{n,ij} (\Delta_{ij}) \mathcal P^\parallel (\bs{r}_{ij}) \cdot \bs{v}_{ij} )}{\bs{r}_{ij}} \notag \\
    & \xrightarrow{\mr{m.e.}} f_{n,ij}'(\Delta_{ij}) \mathcal P^\parallel (\bs{r}_{ij}) - \frac{f_{n,ij}(\Delta_{ij})}{r_{ij}} \mathcal P^\perp (\bs{r}_{ij}) ,
}
where $\mr{m.e.}$ denotes mechanical equilibrium, i.e., all velocities and angular velocities are set to zero.
Likewise we have
\eq{
    \pdv{\bs{F}_{ij}^\parallel}{\bs{v}_{ij}}  & = \eta_{n,ij}(\Delta_{ij}) \mathcal P^\parallel (\bs{r}_{ij}) , \\
    \pdv{\bs{F}_{ij}^\parallel}{\bs{\chi}_{ij}} & = \pdv{\bs{F}_{ij}^\parallel}{\bs{\nu}_{ij}} = 0 .
}

We need the derivatives of $\bs{\xi}_{ij}$ to compute those of $\bs{F}^\perp_{ij}$.
Eq.~\eqref{increment} yields~\cite{Chattoraj2019Oscillatory, Ishima2023Eigenvalue}
\eq{
    \dd \bs{\xi}_{ij} & = \qty ( \mathcal P^\perp (\bs{r}_{ij}) - \frac{\hat{\bs{r}}_{ij} \bs{\xi}_{ij}}{r_{ij}} ) \cdot \dd \bs{r}_{ij} - \dd \bs{\chi}_{ij} \times \hat{\bs{r}}_{ij} ,
}
which leads to
\eq{
    \pdv{\bs{\xi}_{ij}}{\bs{r}_{ij}}  & = \mathcal P^\perp (\bs{r}_{ij}) - \frac{\hat{\bs{r}}_{ij}\bs{\xi}_{ij}}{r_{ij}} , \label{dxidr} \\
    \pdv{\bs{\xi}_{ij}}{\bs{\chi}_{ij}} & = \hat{\bs{r}}_{ij}^\times , \label{dxidtheta} \\
    \pdv{\bs{\xi}_{ij}}{\bs{v}_{ij}} & = \pdv{\bs{\xi}_{ij}}{\bs{\nu}_{ij}} = 0 .
}
These expressions yield the derivatives of the tangential force $\bs{F}^\perp_{ij} = k_{s,ij}(\Delta_{ij}) \bs{\xi}_{ij} + \eta_{s,ij}(\Delta_{ij}) \bs{v}_{ij}^\perp$
\eq{
    \pdv{\bs{F}_{ij}^\perp}{\bs{r}_{ij}} & \xrightarrow{\mr{m.e.}} k_{s,ij}(\Delta_{ij}) \mathcal P^\perp (\bs{r}_{ij}) - \qty ( \frac{k_{s,ij}(\Delta_{ij})}{r_{ij}} + k'_{s,ij}(\Delta_{ij}) ) \hat{\bs{r}}_{ij} \bs{\xi}_{ij} , \\
    \pdv{\bs{F}_{ij}^\perp}{\bs{\chi}_{ij}} & = k_{s,ij}(\Delta_{ij}) \hat{\bs{r}}_{ij}^\times , \\
    \pdv{\bs{F}_{ij}^\perp}{\bs{v}_{ij}} & = \eta_{s,ij} (\Delta_{ij}) \mathcal P^\perp (\bs{r}_{ij}) , \\
    \pdv{\bs{F}_{ij}^\perp}{\bs{\nu}_{ij}} & = \eta_{s,ij} (\Delta_{ij}) \hat{\bs{r}}_{ij}^\times .
}

Finally we compute the derivatives of the pair torque. 
We have
\eq{
    \pdv{\bs{T}_{ij}}{\bs{r}_{ij}} & = \qty ( R_i - \frac{\Delta_{ij}}{2} ) \pdv{\qty ( \hat{\bs{r}}_{ij} \times \bs{F}_{ij}^\perp )}{\bs{r}_{ij}} - \frac{1}{2} \qty ( \hat{\bs{r}}_{ij} \times \bs{F}_{ij}^\perp ) \pdv{\Delta_{ij}}{\bs{r}_{ij}} \notag \\
    & = \qty ( R_i - \frac{\Delta_{ij}}{2} ) \qty ( \hat{\bs{r}}_{ij} \times \pdv{\bs{F}_{ij}^\perp}{\bs{r}_{ij}} - \frac{1}{r_{ij}} \bs{F}_{ij}^\perp \times \mathcal P^\perp (\bs{r}_{ij})  ) - \frac{1}{2} \qty ( \hat{\bs{r}}_{ij} \times \bs{F}_{ij}^\perp ) \qty ( - \hat{\bs{r}}_{ij} ) \notag \\
    & = \qty ( R_i - \frac{\Delta_{ij}}{2} ) \hat{\bs{r}}_{ij} \times \pdv{\bs{F}_{ij}^\perp}{\bs{r}_{ij}} - \bs{F}_{ij}^\perp \times \frac{1}{2} \qty ( \mathcal P^\parallel (\bs{r}_{ij}) + \frac{2R_i-\Delta_{ij}}{r_{ij}} \mathcal P^\perp (\bs{r}_{ij}) ) .
}
Substituting $\bs{F}_{ij}^\perp = k_{s,ij}(\Delta_{ij}) \bs{\xi}_{ij} + \eta_{s,ij}(\Delta_{ij}) \bs{v}_{ij}^\perp$ yields
\eq{
    \pdv{\bs{T}_{ij}}{\bs{r}_{ij}} & \xrightarrow{\mr{m.e.}} \qty ( R_i - \frac{\Delta_{ij}}{2} ) \hat{\bs{r}}_{ij} \times \qty ( k_{s,ij} (\Delta_{ij}) \mathcal P^\perp (\bs{r}_{ij}) - \qty ( \frac{k_{s,ij}(\Delta_{ij})}{r_{ij}} + k'_{s,ij}(\Delta_{ij}) ) \hat{\bs{r}}_{ij} \bs{\xi}_{ij} ) \notag \\
    & \qquad - k_{s,ij} (\Delta_{ij}) \bs{\xi}_{s,ij} \times \frac{1}{2} \qty ( \mathcal P^\parallel (\bs{r}_{ij}) + \frac{2R_i-\Delta_{ij}}{r_{ij}} \mathcal P^\perp (\bs{r}_{ij}) ) \notag \\
    & = \qty ( R_i - \frac{\Delta_{ij}}{2} ) k_{s,ij}(\Delta_{ij}) \hat{\bs{r}}_{ij}^\times - k_{s,ij}(\Delta_{ij}) \bs{\xi}_{ij}^\times \cdot \frac{1}{2} \qty ( \mathcal P^\parallel (\bs{r}_{ij}) + \frac{2R_i-\Delta_{ij}}{r_{ij}} \mathcal P^\perp (\bs{r}_{ij}) ) .
}
Likewise 
\eq{
    \pdv{\bs{T}_{ij}}{\bs{\chi}_{ij}} & = \qty ( R_i - \frac{\Delta_{ij}}{2} ) \hat{\bs{r}}_{ij} \times k_{s,ij}(\Delta_{ij}) \hat{\bs{r}}_{ij}^\times = \qty ( R_i - \frac{\Delta_{ij}}{2} ) \hat{\bs{r}}_{ij} \times k_{s,ij}(\Delta_{ij}) \qty ( I_3 \times \hat{\bs{r}}_{ij} ) = - \qty ( R_i - \frac{\Delta_{ij}}{2} ) k_{s,ij}(\Delta_{ij}) \mathcal P^\perp (\bs{r}_{ij}) , \\
    \pdv{\bs{T}_{ij}}{\bs{v}_{ij}} & = \qty ( R_i - \frac{\Delta_{ij}}{2} ) \hat{\bs{r}}_{ij} \times \eta_{s,ij} (\Delta_{ij}) \mathcal P^\perp (\bs{r}_{ij}) = \qty ( R_i - \frac{\Delta_{ij}}{2} ) \eta_{s,ij} (\Delta_{ij}) \hat{\bs{r}}_{ij}^\times , \\
    \pdv{\bs{T}_{ij}}{\bs{\nu}_{ij}} & = \qty ( R_i - \frac{\Delta_{ij}}{2} ) \hat{\bs{r}}_{ij} \times \eta_{s,ij} (\Delta_{ij}) \hat{\bs{r}}_{ij}^\times = - \qty ( R_i - \frac{\Delta_{ij}}{2} ) \eta_{s,ij} (\Delta_{ij}) \mathcal P^\perp (\bs{r}_{ij}) .
}

In the case of the linear spring-dashpot model introduced in Sec.~\ref{sec: sample preparation}, we obtain
\eq{
    \pdv{\bs{F}_{ij}}{\bs{r}_{ij}} & = k_n \mathcal P^\parallel (\bs{r}_{ij}) + \qty ( k_s - \frac{k_n\Delta_{ij}}{r_{ij}} ) \mathcal P^\perp (\bs{r}_{ij}) - \frac{k_s}{r_{ij}} \hat{\bs{r}}_{ij} \bs{\xi}_{ij} , \label{dFdr linear} \\
    \pdv{\bs{F}_{ij}}{\bs{\chi}_{ij}} & = k_s \hat{\bs{r}}_{ij}^\times , \\
    \pdv{\bs{F}_{ij}}{\bs{v}_{ij}} & = m_{\mr{eff},ij} \gamma_n \mathcal P^\parallel (\bs{r}_{ij}) + m_{\mr{eff},ij} \gamma_s \mathcal P^\perp (\bs{r}_{ij}) , \\
    \pdv{\bs{F}_{ij}}{\bs{\nu}_{ij}} & = m_{\mr{eff},ij} \gamma_s \hat{\bs{r}}_{ij}^\times , \\
    \pdv{\tilde{\bs{T}}_{ij}}{\bs{r}_{ij}} & = \tilde{\Delta}_{ij} k_s \hat{\bs{r}}_{ij}^\times - k_s \bs{\xi}_{s,ij}^\times \cdot \qty ( \frac{1}{2R_i} \mathcal P^\parallel (\bs{r}_{ij}) + \frac{\tilde{\Delta}_{ij}}{r_{ij}} \mathcal P^\perp (\bs{r}_{ij}) ) , \label{dTdr linear} \\
    \pdv{\tilde{\bs{T}}_{ij}}{\bs{\chi}_{ij}} & = - \tilde{\Delta}_{ij} k_s  \mathcal P^\perp (\bs{r}_{ij}) , \\
    \pdv{\tilde{\bs{T}}_{ij}}{\bs{v}_{ij}} & = \tilde{\Delta}_{ij} m_{\mr{eff},ij} \gamma_s \hat{\bs{r}}_{ij}^\times , \\
    \pdv{\tilde{\bs{T}}_{ij}}{\bs{\nu}_{ij}} & = - \tilde{\Delta}_{ij} m_{\mr{eff},ij} \gamma_s \mathcal P^\perp (\bs{r}_{ij}) . \label{dTdw linear}
}
Note that $k_{s,ij}' = 0$ in this case.
The terms involving $\bs{\xi}_{ij}$ and $\tilde \Delta_{ij}$ make the dynamical matrix $H$ non-symmetric.

\section{Simple examples of the spurious instability}\label{sec: minimal model}

\subsection{Disk on a floor}

Consider a two-dimensional disk of mass $m$ and radius $R$ on a floor.
The mass density of the disk is uniform and its moment of inertia is $J = mR^2/2$.
The interaction between the disk and the floor is modeled using the DEM.
For simplicity, the friction coefficient is set to infinity and the dashpot interaction is suppressed.
The position $x$ and the angle $\theta$ of the disk follow the equations of motion
\eq{
    m\dv[2]{x}{t} & = F \coloneqq - k (x + R\theta) , \label{minimal model x} \\
    J\dv[2]{\theta}{t} & = T \coloneqq (R-\Delta) F , \label{minimal model theta}
}
where $k$ is the spring constant of the tangential spring and $\Delta<R$ is the overlap between the disk and the floor.
The contact never slips because the friction coefficient is infinite.
Although the motion of the disk along the $y$-axis determines $\Delta$, it remains fixed and does not need to be considered provided that the force normal to the surface is initially balanced.
Note that we only have $F=T=0$ in the case of the exact Coulomb friction law.

Eqs.~\eqref{minimal model x} and \eqref{minimal model theta} are already linear in $x$ and $\theta$, and the time evolution of the mechanical energy is immediately given by
\eq{
    \dv{t} \qty ( \frac{1}{2}mv^2 + \frac{1}{2}J\omega^2 + \frac{1}{2}k \begin{pmatrix} x & R\theta \end{pmatrix} \cdot
    \begin{pmatrix}
        1 & 1 - \alpha/2 \\
        1 - \alpha/2 & 1-\alpha
    \end{pmatrix} \cdot
    \begin{pmatrix} x \\ R\theta \end{pmatrix} ) & = -\frac{k\alpha}{2} \begin{pmatrix} v & R\omega \end{pmatrix} \cdot
    \begin{pmatrix}
        0 & 1 \\
        -1 & 0
    \end{pmatrix} \cdot
    \begin{pmatrix} x \\ R\theta \end{pmatrix}
    , \label{minimal model energy}
}
where $\alpha \coloneqq \Delta/R$.
To proceed, we set $m=k=R=1$ and rewrite Eqs.~\eqref{minimal model x} and \eqref{minimal model theta} in the form given by Eq.~\eqref{linear eom} with\footnote{
In contrast to the main text, we do not adopt the normalization introduced in Eq.~\eqref{change}. 
Accordingly, we use notation that differs slightly from that in the main text, i.e., the matrices are not denoted with tildes.
}
\eq{
    \mathcal A & =
    \begin{pmatrix}
        O_2 & M^{-1} \\
        H & O_2
    \end{pmatrix} , \\
    M & =
    \begin{pmatrix}
        1 & 0 \\
        0 & 1/2
    \end{pmatrix} , \\
    H & =
    \begin{pmatrix}
        -1 & -1 \\
        - 1 + \alpha & -1 + \alpha
    \end{pmatrix} .
}
The state vector is defined as $\bs{q} \coloneqq \begin{pmatrix} x & \theta & v & \omega \end{pmatrix}^T$.
Likewise we can define the energy and dissipation functions
\eq{
    E (\bs{q}) & \coloneqq \frac{1}{2} \bs{q} \cdot \mathcal E \cdot \bs{q} , \\
    R (\bs{q}) & \coloneqq \frac{1}{2} \bs{q} \cdot \mathcal R \cdot \bs{q} ,
}
where
\eq{
    \mathcal{E} & \coloneqq 
    \begin{pmatrix}
        -H_{\mr{sym}} & O_2 \\
        O_2 & M
    \end{pmatrix}
    = 
    \begin{pmatrix}
        1 & 1 - \alpha/2 & 0 & 0 \\
        1 - \alpha/2 & 1-\alpha & 0 & 0 \\
        0 & 0 & 1 & 0 \\
        0 & 0 & 0 & 1/2
    \end{pmatrix} , \\
    \mathcal{R} & \coloneqq 
    \begin{pmatrix}
        O_2 & H_{\mr{skew}} \\
        - H_{\mr{skew}} & O_2
    \end{pmatrix}
    = 
    \frac{\alpha}{2}
    \begin{pmatrix}
        0 & 0 & 0  & -1 \\
        0 & 0 & 1 & 0 \\
        0 & 1 & 0 & 0 \\
        -1 & 0 & 0 & 0         
    \end{pmatrix}
    .
}

The eigenvalues of $\mathcal{E}$ are
\eq{
    \lambda_\pm & \coloneqq 1 - \frac{\alpha}{2} \pm \sqrt{\qty ( \frac{\alpha}{2} )^2 + \qty ( 1 - \frac{\alpha}{2} )^2} ,\,\,\, 1,\,\,\, 1/2.
}
The inequality $\lambda_-<0$ holds for all $0<\alpha<1$ and thus the disk is on a saddle point of the potential energy, contrary to our intuition.
The eigenvalues of $\mathcal R$ are $\pm1$, indicating that the energy is both supplied and dissipated via the tangential spring.

We set $\alpha = 1/2$ and solve Eqs.~\eqref{minimal model x} and \eqref{minimal model theta} with initial conditions $x = \omega = 1$ and $\theta = v = 0$.
In this case, the equations of motion simplify to
\eq{
    \dv[2]{x}{t} & = - x - \theta , \\
    \dv[2]{\theta}{t} & = - x - \theta ,
}
which immediately yield
\eq{
    x - \theta & = -t + 1 , \\
    x + \theta & = \frac{1}{\sqrt{2}} \sin \qty ( \sqrt{2}t ) + \cos \qty ( \sqrt{2}t ) \eqqcolon u .
}
In the long-time limit, the mechanical energy is given by
\eq{
    E (\bs{q}) & \xrightarrow{t\to\infty} - \frac{1}{8}ut .
}
Since $u$ is an oscillatory function and remains of order unity, the amplitude of the energy grows linearly with time, which is a direct consequence of the spurious instability of the DEM discussed in Sec.~\ref{sec: linear stability and mechanical energy}.

We note that the skew-symmetric part $H_{\mr{skew}}$ arises from the inconsistency between the definitions of the force and torque $\Delta$.
In the definition of $F$, the force due to rotation is defined as $-kR\theta$, meaning that it acts on the surface of the disk while in the definition of $T$, the lever arm is defined as $R-\Delta$ with the deformation of the disk taken into account.
This non-symmetry is thus eliminated by changing the force $-kR\theta$ to $-k(R-\Delta)\theta$ or the lever arm $R-\Delta$ to $R$.
We adopt the former modification in the example presented in the following subsection.
See also Appendix~\ref{sec: alternative torque} for a detailed discussion.

\subsection{Disk between two perpendicular walls}\label{sec: two walls}

Another example is a two-dimensional disk in contact with two mutually perpendicular walls.
This can be viewed as a two-dimensional pseudo-rattler.
The walls are aligned with the $x$- and $y$-axes and the disk is initially placed at $(X, Y)$, where $X,Y>0$.
The tangential displacements of the two walls are measured relative to $(X_w,0)$ and $(0,Y_w)$, respectively, and if $X\neq X_w$ or $Y\neq Y_w$, tangential forces are already present in the initial state.
All spring constants in the interactions between the disk and the walls are set to $k$ and we change the definition of the force as mentioned at the end of the previous subsection.
The other properties of the disk is the same as in the previous subsection.
The force and torque balance requires that the initial state satisfies
\eq{
    0 & = k ( X_w - X ) + k ( R - X ) , \label{force balance x} \\
    0 & = k ( Y_w - Y ) + k ( R - Y ) , \label{force balance y} \\
    0 & = ( R - Y ) k ( X_w - X ) - ( R - X ) k ( Y_w - Y ) . \label{torque balance}
}
The equations of motion are given by
\eq{
    m \dv[2]{x}{t} & = k(R-X-x) + k(X_w-X-x-(R-Y-y)\theta) , \\
    m \dv[2]{y}{t} & = k(R-Y-y) + k(Y_w-Y-y+(R-X-x)\theta) , \\
    J \dv[2]{\theta}{t} & = (R-Y-y) k (X_w-X-x-(R-Y-y)\theta) - (R-X-x) k (Y_w-Y-y-(R-X-x)\theta) .
}
Using Eqs.~\eqref{force balance x}--\eqref{torque balance} and assuming $x$, $y$, and $\theta$ are infinitesimal, we have
\eq{
    m \dv[2]{x}{t} & = -2kx - k(R-Y)\theta , \label{two walls x} \\
    m \dv[2]{y}{t} & = -2ky + k(R-X)\theta , \label{two walls y} \\
    J \dv[2]{\theta}{t} & = -k(R-Y_w) x + k(R-X_w)y + k((R-X)^2-(R-Y)^2)\theta . \label{two walls theta}
}
Setting $m = k = R = 1$, $X = Y = \alpha_0$, and $X_w = Y_w = \alpha_0+\alpha$, Eqs.~\eqref{force balance x}--\eqref{torque balance} reduce to $\alpha_0 = 1+\alpha$ and Eqs.~\eqref{two walls x}--\eqref{two walls theta} simplify to
\eq{
    \dv[2]{x}{t} & = -2x + \alpha\theta , \label{two walls x normalized} \\
    \dv[2]{y}{t} & = -2y - \alpha\theta , \\
    \frac{1}{2}\dv[2]{\theta}{t} & = 2\alpha x -2\alpha y , \label{two walls theta normalized}
}
which leads to
\eq{
    M & = 
    \begin{pmatrix}
        1 & 0 & 0 \\
        0 & 1 & 0 \\
        0 & 0 & 1/2
    \end{pmatrix} , \\
    H & = 
    \begin{pmatrix}
        -2 & 0 &  \alpha \\
        0 & -2 & -\alpha \\
        2\alpha & -2\alpha & 0         
    \end{pmatrix}
    .
}
The energy and dissipation matrices are given by
\eq{
    \mathcal E & =
    \begin{pmatrix}
        -H_{\mr{sym}} & O_3 \\
        O_3 & M
    \end{pmatrix} , \\
    \mathcal R & = 
    \begin{pmatrix}
        O_3 & H_{\mr{skew}} \\
        -H_{\mr{skew}} & O_3
    \end{pmatrix} .
}
The eigenvalues of $\mathcal E$ are $1/2,1,2,1\pm\sqrt{1+9\alpha^2/2}$ and those of $\mathcal R$ are $\pm\sqrt{2}\alpha/2,0$, implying that this disk is on a saddle point and exhibits the spurious instability.
In this example, $H_{\mr{skew}}$ arises unless $\alpha = X_w - X = Y_w - Y = 0$, i.e., the non-symmetric dynamical matrix results from the residual tangential displacements in the initial state.

\section{Another definition of the pair torque}\label{sec: alternative torque}

\begin{table}[h]
    \renewcommand{\arraystretch}{1.2}
    \setlength\tabcolsep{0.5em}
    \centering
    \caption{
    Corresponding table to Table~\ref{tab:samples}, computed using the alternative definition of pair torque in Appendix~\ref{sec: alternative torque}.
    }
    \begin{tabular}{ccc} \hline \hline
        $\mu$ & $\phi$ & $N=128$ \\ \hline
        1.00  & 0.570  &  4      \\
              & 0.585  &  8      \\
              & 0.600  & 10      \\ \hline \hline
    \end{tabular}
    \label{tab:samples_torque}
\end{table}

\begin{figure}[ht]
    \centering

    \begin{minipage}{0.32\linewidth}
        \centering
        \includegraphics[width=\linewidth]{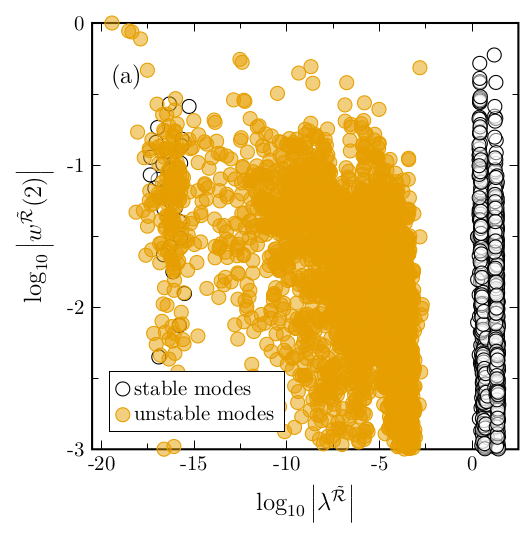}
    \end{minipage}
    \begin{minipage}{0.32\linewidth}
        \centering
        \includegraphics[width=\linewidth]{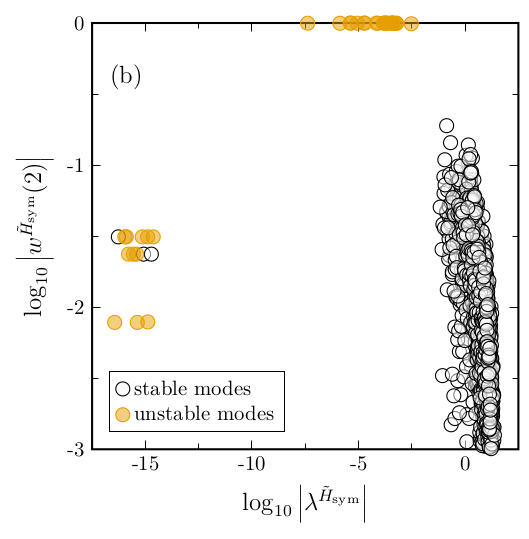}
    \end{minipage}
    \begin{minipage}{0.32\linewidth}
        \centering
        \includegraphics[width=\linewidth]{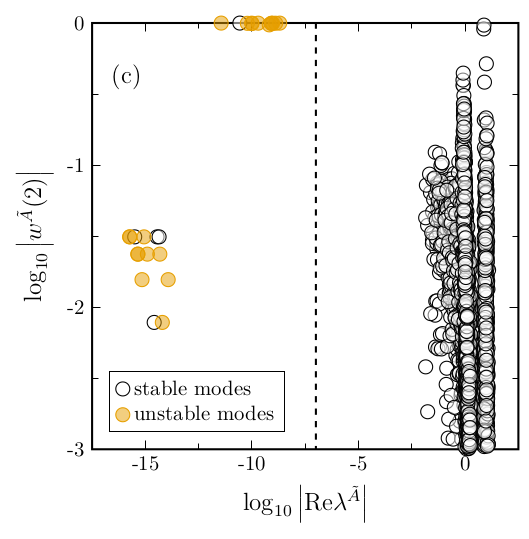}
    \end{minipage}
    \caption{
    Corresponding figure to Fig.~\ref{fig:rattler weight}, computed using the alternative definition of pair torque in Appendix~\ref{sec: alternative torque}.
    }
    \label{fig:rattler weight torque}
\end{figure}

As mentioned in Appendix~\ref{sec: minimal model}, there are several versions of definitions of the tangential force and the torque~\cite{Cundall1979A,Luding2008Cohesive,Plimpton1995Fast,Thompson2022LAMMPS,Chattoraj2019Oscillatory,Chattoraj2019Noise}.
The differences arise from whether the deformation of particles is taken into account when defining their contact point.
In the main text, we define the lever arm of the pair torque as $R_i-\Delta_{ij}/2$ in Eq.~\eqref{pair torque} while we use $\bs{\chi}_{ij}$ and $\bs{\nu}_{ij}$ to define $\bs{F}^\perp_{ij}$, in which only the radii $R_i$ and $R_j$ appear without the correction of the overlap $\Delta_{ij}$.
This definition is adopted in LAMMPS~\cite{Plimpton1995Fast,Thompson2022LAMMPS}.

However, a simpler definition is often adopted in the literature~\cite{Cundall1979A,Chattoraj2019Oscillatory,Chattoraj2019Noise}, and here we verify that the main features of pseudo-rattlers are insensitive to the choice of definition.
In this alternative definition, the lever arm is simply taken to be $R_i$, i.e., the overlap is ignored. 
Accordingly, $\tilde{\Delta}_{ij}$ is replaced by unity in Eqs.~\eqref{dTdr linear}--\eqref{dTdw linear}, and the term involving $\mathcal P^\parallel (\bs{r}_{ij})$ disappears from Eq.~\eqref{dTdr linear}.

Fig.~\ref{fig:rattler weight torque} shows the corresponding results to those in Fig.~\ref{fig:rattler weight}, computed using this simpler definition. 
The figure demonstrates that these results are qualitatively unchanged and that pseudo-rattlers still contribute to the saddle directions of the potential energy defined by $-\tilde H_{\mr{sym}}$, further supporting the conclusion of Sec.~\ref{sec: energy matrix and rattlers}.
See also Table~\ref{tab:samples_torque}, which should be compared with Table~\ref{tab:samples}.

We finally mention another definition of the tangential force~\cite{Luding2008Cohesive}.
In Ref.~\cite{Luding2008Cohesive}, the radii $R_i$ and $R_j$ in $\bs{\chi}_{ij}$ and $\bs{\nu}_{ij}$ were replaced by $R_i-\Delta_{ij}/2$ and $R_j-\Delta_{ij}/2$, in addition to the lever arm of the pair torque.
This definition is essentially equivalent to that in Appendix~\ref{sec: two walls} and the qualitative behavior does not change regardless of the definitions even in this case.

\clearpage

\section{Additional numerical results}\label{sec: additional numerical results}

Fig.~\ref{fig:mu_000} shows a scatter plot of $P^{(r)}$ versus $\omega$ for $N=1024$, $\mu=0.00$, and $\phi=0.690$.
In frictionless systems, $W_i^{(r)}=1$ and $P_i^{(\theta)}=0$ hold for all $i$ because the rotational degrees of freedom are not relevant.
In Figs.~\ref{fig:mu_001} and \ref{fig:mu_100}, we show the counterparts of Fig.~\ref{fig:mu_020} for $\mu=0.01$ and $\mu=1.00$, respectively.
These figures demonstrate that the intermediate- and high-frequency regimes are qualitatively independent of the friction coefficient and the packing fraction.

Likewise, Figs.~\ref{fig:spectra_N128}--\ref{fig:mu_100_N128} are the counterparts of Figs.~\ref{fig:spectra}, \ref{fig:mu_020}, and \ref{fig:mu_000}--\ref{fig:mu_100}, respectively, for $N=128$.
It is clear that the qualitative behavior discussed in the main text is already evident at this relatively small system size, indicating that the system size dependence is sufficiently weak for the present purpose.

\begin{figure}[ht]
    \centering
    \begin{minipage}{0.32\linewidth}
        \centering
        \includegraphics[width=\linewidth]{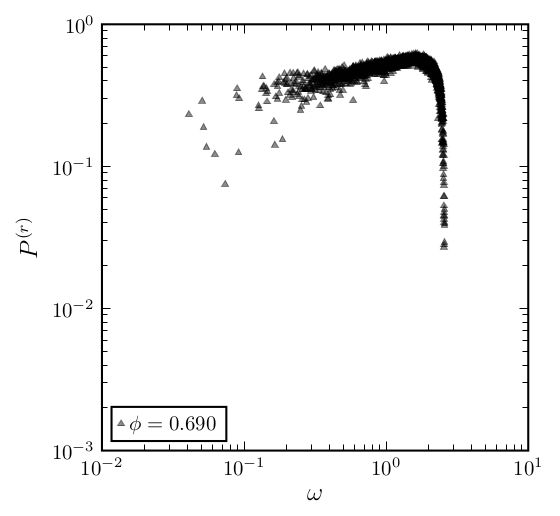}
    \end{minipage}
    \caption{
    Scatter plot of $P^{(r)}$ versus $\omega$ for $N=1024$, $\mu=0.00$, and $\phi=0.690$.
    }
    \label{fig:mu_000}
\end{figure}

\begin{figure}[b]
    \centering

    \begin{minipage}{0.32\linewidth}
        \centering
        \includegraphics[width=\linewidth]{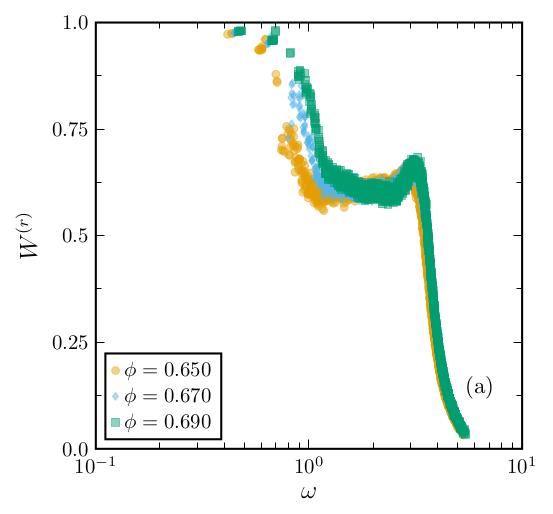}
    \end{minipage}
    \begin{minipage}{0.32\linewidth}
        \centering
        \includegraphics[width=\linewidth]{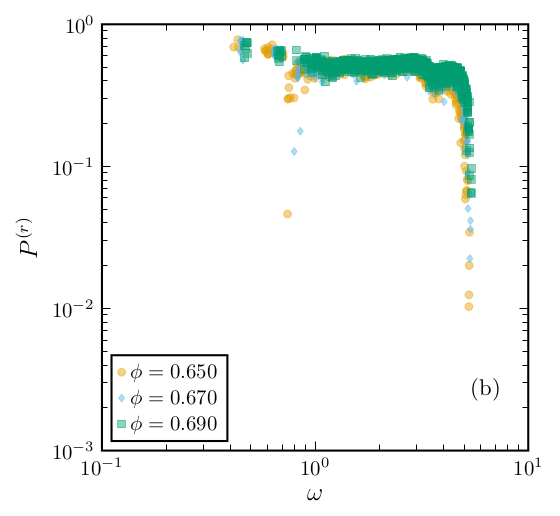}
    \end{minipage}
    \begin{minipage}{0.32\linewidth}
        \centering
        \includegraphics[width=\linewidth]{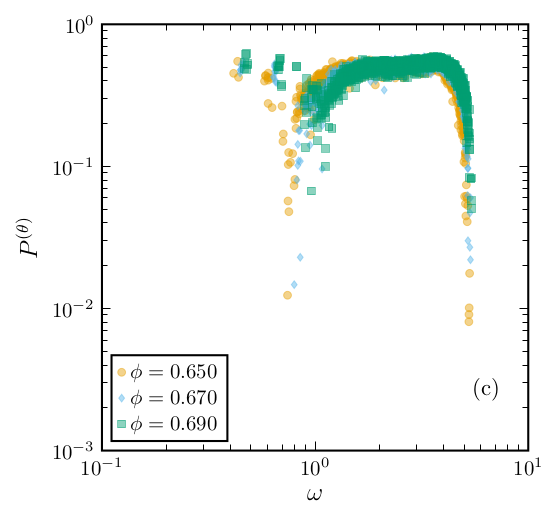}
    \end{minipage}
    \caption{
    Scatter plots of $W^{(r)}$, $P^{(r)}$, and $P^{(\theta)}$ versus $\omega$ in panels (a), (b), and (c), respectively.
    These results were obtained using 10 samples with $N=1024$ and $\mu=0.01$.
    }
    \label{fig:mu_001}
\end{figure}

\begin{figure}
    \centering

    \begin{minipage}{0.32\linewidth}
        \centering
        \includegraphics[width=\linewidth]{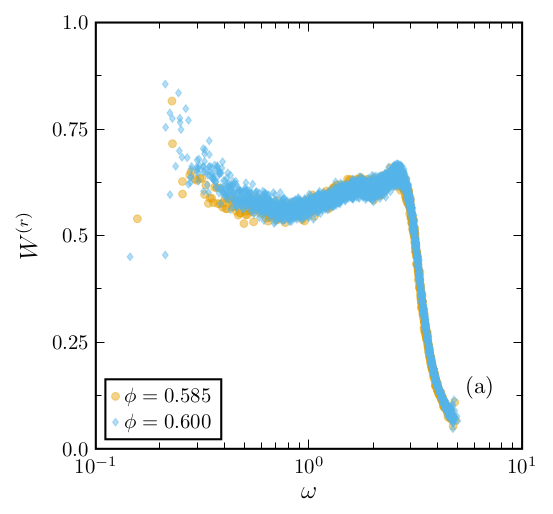}
    \end{minipage}
    \begin{minipage}{0.32\linewidth}
        \centering
        \includegraphics[width=\linewidth]{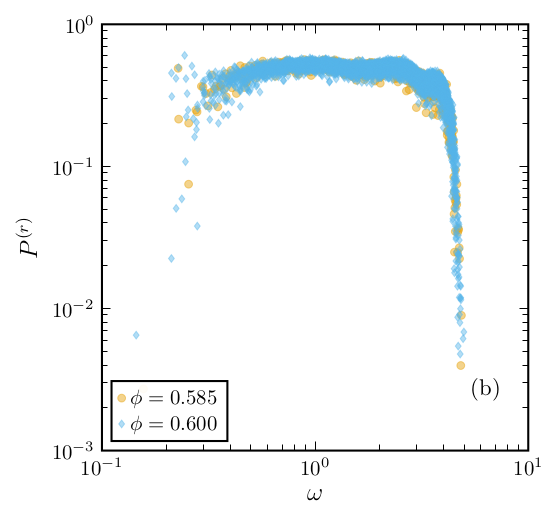}
    \end{minipage}
    \begin{minipage}{0.32\linewidth}
        \centering
        \includegraphics[width=\linewidth]{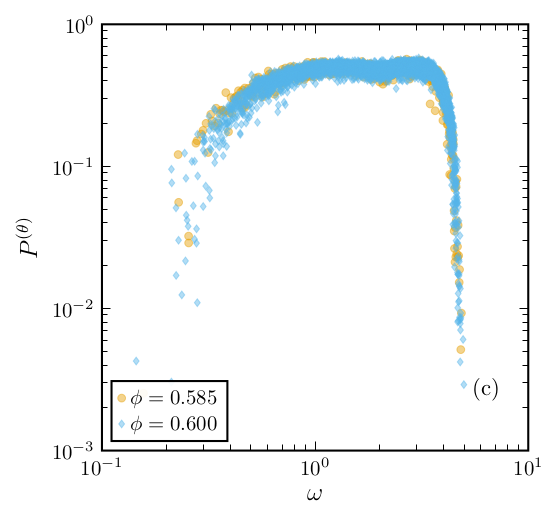}
    \end{minipage}
    \caption{
    Scatter plots of $W^{(r)}$, $P^{(r)}$, and $P^{(\theta)}$ versus $\omega$ in panels (a), (b), and (c), respectively.
    These results were obtained using two samples at $\phi=0.585$ and nine samples at $\phi=0.600$ with $N=1024$ and $\mu=1.00$.
    }
    \label{fig:mu_100}
\end{figure}

\begin{figure}
    \centering

    \begin{minipage}{0.32\linewidth}
        \centering
        \includegraphics[width=\linewidth]{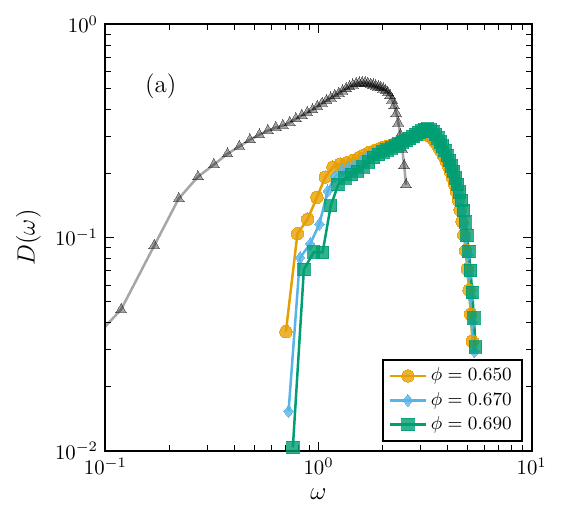}
    \end{minipage}
    \begin{minipage}{0.32\linewidth}
        \centering
        \includegraphics[width=\linewidth]{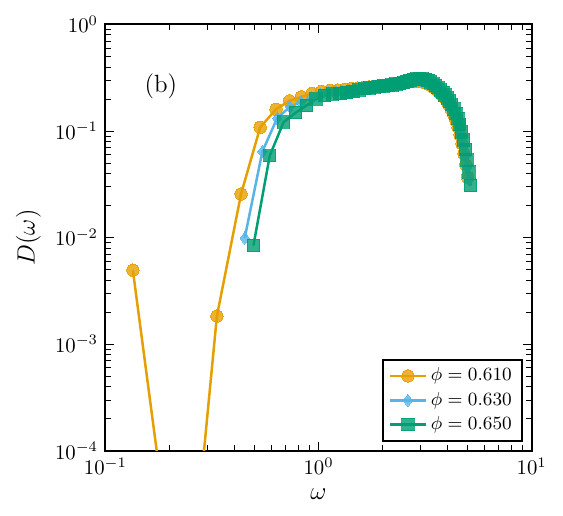}
    \end{minipage}
    \begin{minipage}{0.32\linewidth}
        \centering
        \includegraphics[width=\linewidth]{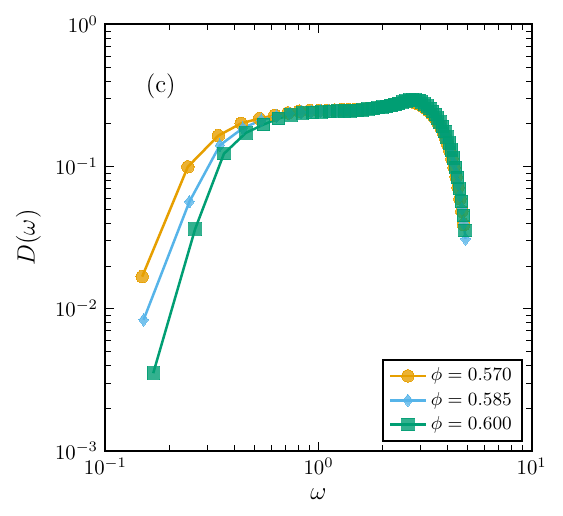}
    \end{minipage}
    \caption{
    Densities of states of the energy matrix $\tilde H_{\mr{sym}}$ for $\mu=0.01, 0.20,$ and $1.00$ in panels (a), (b), and (c), respectively.
    In panel (a), the density of states for $\mu=0.00$ and $\phi=0.69$ is also shown as black triangles for comparison.
    These results were obtained using samples with $N=128$.
    }
    \label{fig:spectra_N128}
\end{figure}

\begin{figure}
    \centering

    \begin{minipage}{0.32\linewidth}
        \centering
        \includegraphics[width=\linewidth]{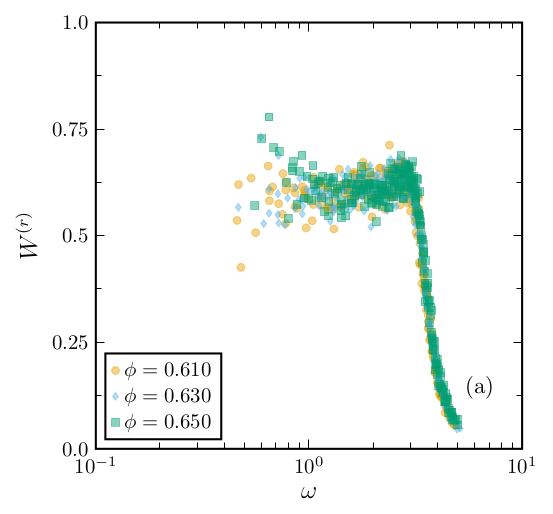}
    \end{minipage}
    \begin{minipage}{0.32\linewidth}
        \centering
        \includegraphics[width=\linewidth]{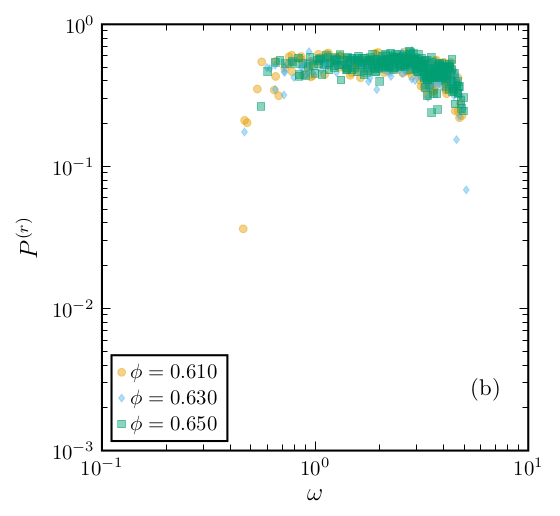}
    \end{minipage}
    \begin{minipage}{0.32\linewidth}
        \centering
        \includegraphics[width=\linewidth]{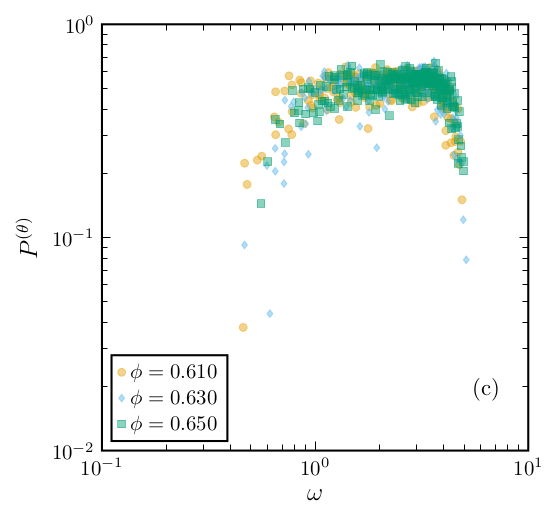}
    \end{minipage}
    \caption{
    Scatter plots of $W^{(r)}$, $P^{(r)}$, and $P^{(\theta)}$ versus $\omega$ in panels (a), (b), and (c), respectively.
    These results were obtained using 10 samples with $N=128$ and $\mu=0.20$.
    }
    \label{fig:mu_020_N128}
\end{figure}

\begin{figure}
    \centering
    \begin{minipage}{0.32\linewidth}
        \centering
        \includegraphics[width=\linewidth]{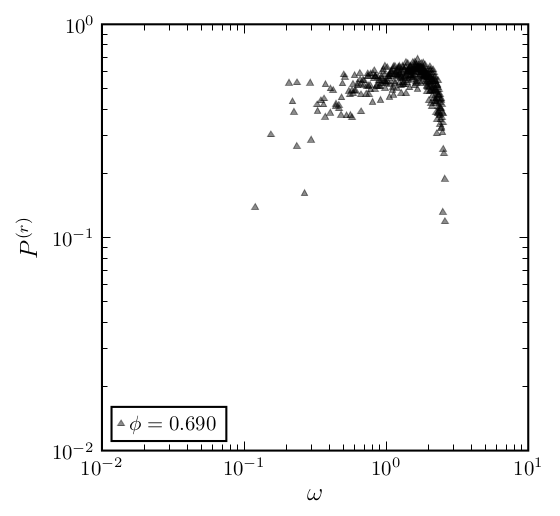}
    \end{minipage}
    \caption{
    Scatter plot of $P^{(r)}$ versus $\omega$ for $N=128$, $\mu=0.00$, and $\phi=0.690$.
    }
    \label{fig:mu_000_N128}
\end{figure}

\begin{figure}
    \centering

    \begin{minipage}{0.32\linewidth}
        \centering
        \includegraphics[width=\linewidth]{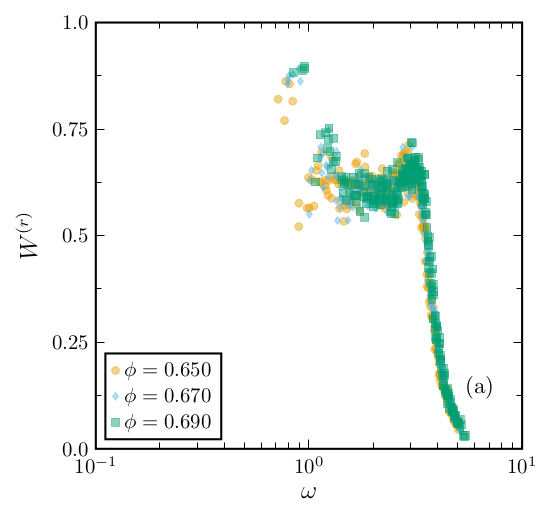}
    \end{minipage}
    \begin{minipage}{0.32\linewidth}
        \centering
        \includegraphics[width=\linewidth]{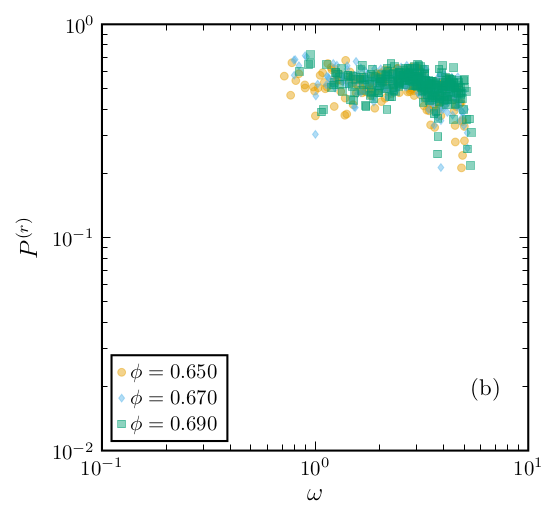}
    \end{minipage}
    \begin{minipage}{0.32\linewidth}
        \centering
        \includegraphics[width=\linewidth]{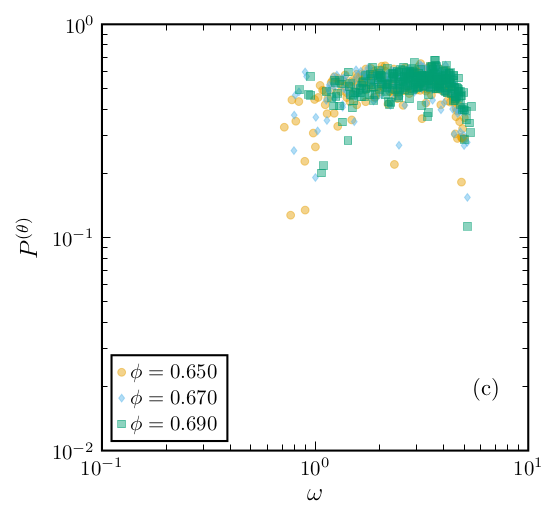}
    \end{minipage}
    \caption{
    Scatter plots of $W^{(r)}$, $P^{(r)}$, and $P^{(\theta)}$ versus $\omega$ in panels (a), (b), and (c), respectively.
    These results were obtained using 10 samples with $N=128$ and $\mu=0.01$.
    }
    \label{fig:mu_001_N128}
\end{figure}

\begin{figure}
    \centering

    \begin{minipage}{0.32\linewidth}
        \centering
        \includegraphics[width=\linewidth]{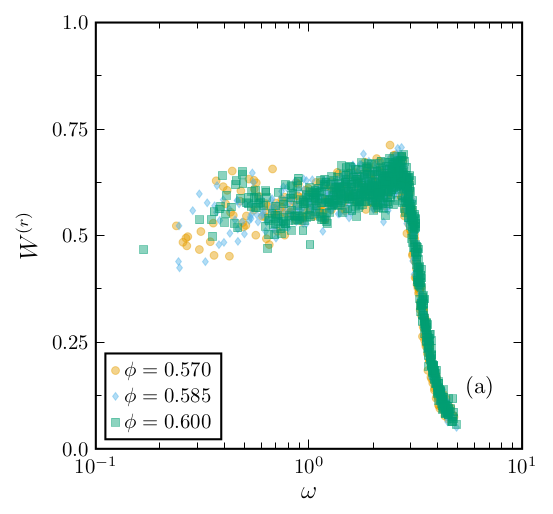}
    \end{minipage}
    \begin{minipage}{0.32\linewidth}
        \centering
        \includegraphics[width=\linewidth]{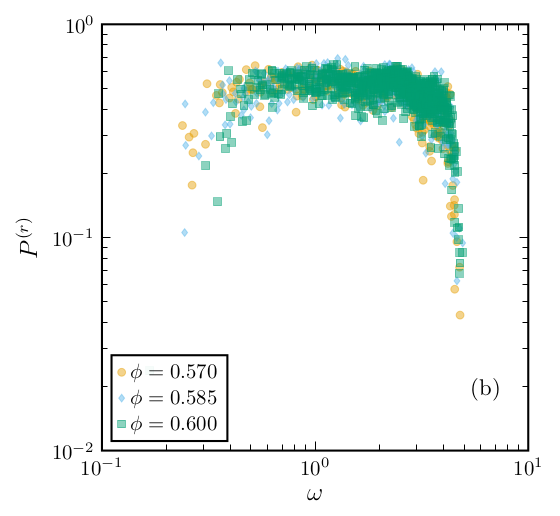}
    \end{minipage}
    \begin{minipage}{0.32\linewidth}
        \centering
        \includegraphics[width=\linewidth]{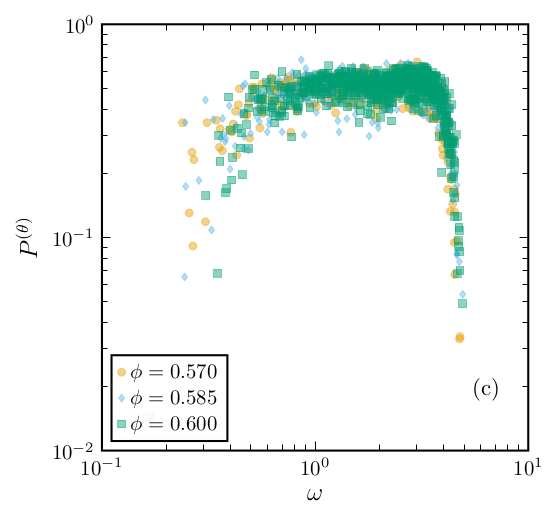}
    \end{minipage}
    \caption{
    Scatter plots of $W^{(r)}$, $P^{(r)}$, and $P^{(\theta)}$ versus $\omega$ in panels (a), (b), and (c), respectively.
    These results were obtained using 7 samples at $\phi=0.570$, 9 samples at $\phi=0.585$, and 10 samples at $\phi=0.600$, with $N=128$ and $\mu=1.00$.
    }
    \label{fig:mu_100_N128}
\end{figure}

\clearpage

\end{document}